\documentclass[12pt]{article}
\usepackage{amsmath}
\usepackage{graphicx}
\usepackage{enumerate}
\usepackage{natbib}
\usepackage{url} 

\usepackage{amsthm}
\usepackage{verbatim,amssymb,epsfig,enumitem,amsbsy,enumitem,color,bbm,mathtools}
\usepackage{hyperref}
\usepackage{subfigure}
\usepackage{multirow}
\usepackage{geometry}
\usepackage{dsfont}
\usepackage{appendix}
\usepackage{algorithm}
\usepackage{algorithmic}
\usepackage[modulo]{lineno}

\def\cov{\mathrm{cov}}

\def\t{\mathrm{T}}
\def\real{\mathbb{R}}
\def\tdomain{\mathcal{T}}
\def\bx{\mathbf{x}}
\def\by{\mathbf{y}}
\def\bu{\mathbf{u}}

\def\bX{\mathbf{X}}

\def\bepsilon{\boldsymbol{\epsilon}}
\def\bmu{\boldsymbol{\mu}}

\def\mvec{\mathrm{vec}}

\newtheorem{theorem}{Theorem}
\newtheorem{lemma}{Lemma}
\newtheorem{corollary}{Corollary}

\newtheorem{remark}{Remark}

\newtheorem{assumption}{Assumption}

\newcommand{\blind}{0}

\begin{document}

\def\spacingset#1{\renewcommand{\baselinestretch}%
{#1}\small\normalsize} \spacingset{1}


\if0\blind
{
  \title{\bf Dynamic Principal Component Analysis in High Dimensions}
  \author{Xiaoyu Hu \\
  Department of Statistics \& Data Science, \\
  National University of Singapore, Singapore\\
  Fang Yao \\
    Department of Probability \& Statistics, School of Mathematical Sciences,\\
    Center for Statistical Science, Peking University, Beijing, China}
  \maketitle
} \fi

\if1\blind
{
  \bigskip
  \bigskip
  \bigskip
  \begin{center}
    {\LARGE\bf Dynamic Principal Component Analysis in High Dimensions}
\end{center}
  \medskip
} \fi

\bigskip
\begin{abstract}
Principal component analysis is a versatile tool to reduce dimensionality which has wide applications in statistics and machine learning. It is particularly useful for modeling data in high-dimensional scenarios where the number of variables $p$ is comparable to, or much larger than the sample size $n$.
Despite an extensive literature on this topic, researchers have focused on modeling static principal eigenvectors, which are not suitable for stochastic processes that are dynamic in nature.  
To characterize the change in the entire course of high-dimensional data collection, we propose a unified framework to {directly estimate dynamic eigenvectors of covariance matrices}. 
Specifically, we formulate an optimization problem by combining the local linear smoothing and regularization penalty together with the orthogonality constraint, which can be effectively solved by manifold optimization algorithms. We show that our method is suitable for high-dimensional data observed under both common and irregular designs, and theoretical properties of the estimators are investigated under $l_q (0 \leq q \leq 1)$ sparsity. 
Extensive experiments demonstrate the effectiveness of the proposed method in both simulated and real data examples.
\end{abstract}

\noindent%
\vfill

\newpage
\spacingset{1.9} 

\section{Introduction}\label{sec:intro}

Principal component analysis (PCA) has been widely used to reduce dimensionality and extract useful features by transforming the original variables into a few new uncorrelated variables while retaining most information in the data \citep{anderson1963asymptotic}. It is an important tool in various applications, such as data compression and reconstruction \citep{sirovich1987low,turk1991face}.
Despite its importance, existing works in this field mainly focus on modeling a static decomposition, where principal eigenvectors are invariant with respect to time. 
However, technological advances enable data collection in dynamic environments that often vary with time or other index variables. Such data are expected to possess dynamic structures with eigenvectors/eigenspaces varying with time, which makes existing methods less applicable.
To tackle this issue, we aim to dynamically estimate leading eigenvectors of covariance matrices to capture the time-varying information, which is referred to as the dynamic PCA (DPCA). It has wide applications in signal processing, e.g., subspace tracking \citep{delmas2010subspace}.

A straightforward way to conduct the DPCA is to perform PCA on the sample covariance matrix at each observed grid point \citep{berrendero2011principal}. However, this has several limitations. First, one cannot obtain a smooth estimate over the whole time period in an integrative manner by their method. Second, under the irregularly/sparsely observed case for each subject, one cannot directly calculate the sample covariance matrix, which makes their method inapplicable. More importantly, PCA is known to behave poorly in high-dimensional settings, where the number of variables $p$ is comparable to or much larger than the sample size $n$ \citep{johnstone2009consistency}.

In the high-dimensional static case, various sparse PCA methods are developed and studied in the literature \citep{jolliffe2003modified,zou2006sparse,shen2008sparse,amini2008high,witten2009penalized,berthet2013computational,brennan2019optimal}.
For spiked covariance models, \citet{johnstone2009consistency} proposed the diagonal thresholding algorithm by retaining variables with large sample variances, and \citet{ma2013sparse,deshpande2014sparse,krauthgamer2015semidefinite} further refined this estimation. Moreover, \citet{ding2019subexponential,holtzman2020greedy} proposed different algorithms for the sparse PCA problem in the spiked covariance model. 
In more general settings, \citet{vu2013minimax} gave the non-asymptotic lower and upper bounds on the minimax subspace estimation error. In addition, \citet{vu2013fantope} considered a convex relaxation strategy based on the convex hull of low rank projection matrices.
Yet, if one adopts the method in \citet{berrendero2011principal}, the existing sparse PCA methods in the static case are not readily applicable for modeling time-varying principal eigenvectors. For instance, the irregular sampling scheme and the dependence among measurements from the same subject pose new challenges to methodological and theoretical developments. 

Another possible way to obtain time-varying eigenvectors is through the eigendecomposition of estimated dynamic covariance matrices, and dynamic covariance models have been explored in the literature.
In low-dimensional settings, the nonparametric or semiparametric estimators are constructed  \citep{zhu2009intrinsic,yin2010nonparametric,yuan2012local}. For high-dimensional data, \citet{chen2016dynamic} proposed a sparse estimate using the kernel smoothing and thresholding. 
However, it is unclear about the quality of eigenvector estimates based on performing standard PCA on the sparse covariance matrices.
Our numerical studies reveal that the principal eigenvectors deduced from dynamic covariance estimation in \citet{chen2016dynamic} perform sub-optimally, especially when the observational grids are sparse and the dimension is large.
Moreover, these existing works about dynamic covariance matrices did not account for the dependence among observations from the same process which is an important nature of repeated measurements \citep{cai2011optimal}. 

In this work, we propose a unified framework with theoretical guarantees for the DPCA. Specifically, to resolve the problems caused by the sample covariance matrix and the irregular sampling scheme, we use the local linear smoothing \citep{jianqing1996local} for estimation. To deal with high dimensionality, we restrict our attention to the eigenvectors with sparsity structures. 
Consequently, we formulate an optimization problem which combines the local linear smoothing and sparse regularization.
The proposed method has some remarkable features. 
First, it is applicable to high-dimensional data under both common and irregular/sparse designs \citep{cai2011optimal}.
Second, instead of adopting the convex relaxation strategy which is computationally expensive with the computational cost $O(p^3)$ per iteration \citep{vu2013fantope}, our optimization problem is directly defined on the Stiefel manifold, which can be solved by leveraging manifold optimization algorithms, e.g., the proximal gradient method in \citet{chen2020proximal} operating with $O(p^2d)$, where $d$ is the number of principal eigenvectors of interest. 
Third, our procedure consists of two steps: the first step generates an initial estimate from the optimization, and the second step refines the estimate by hard thresholding and re-optimization on the reduced set of variables.
This two-step algorithm helps successfully identify significant features and enhance the interpretability, which leads to consistent estimators under the $l_q$ ($0 \leq q \leq 1$) sparsity. 
Moreover, we show that the convergence rate of resulting estimators exhibits a phase transition phenomenon that attains either nonparametric or parametric rate, depending on the sampling frequency, i.e., how sparse/dense the repeated measurements are observed, see Section \ref{sec:thm}. It is noteworthy that the convergence of estimated principal eigenvectors is faster than that in dynamic covariance estimation, which coincides with findings in the static case \citep{vu2012minimax, bickel2008covariance, cai2012minimax}.

While both DPCA and functional PCA (FPCA) \citep{ramsay2005}  are tools to model random functions, they are essentially different frameworks, see Remark \ref{rmk:dpca-fpca}.
Although FPCA is widely used to represent a single or a small  number of functional processes, its performance is not guaranteed and can be unreliable in high dimensions due to error accumulation \citep{yao2005functionala, yao2005functionalb, chiou2014multivariate}. In this regard, the DPCA is preferred to capture the dynamic information with low-dimensional structures in applications such as data compression and reconstruction. 
This is illustrated in the real data example in terms of recovery errors in Section \ref{sec:realdata}.

The remainder of the article is organized as follows. In Section \ref{sec:dpca}, we first introduce the dynamic PCA, then we provide the $l_q (0 \le q \le 1)$ sparsity assumption in dynamic settings and the formulation of our optimization problem, and describe procedures for practical implementation. In Section \ref{sec:thm}, we present theoretical results under suitable regularity conditions. Simulation results are included in Section \ref{sec:sim}, followed by an application to the heartbeat sound data in Section \ref{sec:realdata}. The additional results and technical proofs are deferred to the Appendix and Supplementary Material.

\section{Dynamic principal component analysis with sparsity}\label{sec:dpca}

\subsection{Dynamic principal component analysis}\label{subsec:dps}
We begin with some notations used in the sequel.
For a matrix $A = (a_{ij})_{i,j=1}^p \in \real^{p \times p}$, $\mvec(A)$ denotes the vector with length $p^2$ obtained by stacking the columns of $A$. We define the Frobenius norm $\|A\|_F = \left(\sum_{i,j} a_{ij}^2\right)^{1/2}$, the elementwise $l_{\infty}$ norm $\|A\|_{\infty} = \max_{1\le i,j \le p} |a_{ij}|$ and the elementwise $l_1$ norm $\|A\|_1 = \sum_{i,j} |a_{ij}|$. For a vector $\bu \in \real^p$, denote its $l_q$ norm by $\|\bu\|_q = \left(\sum_{j=1}^p |u_j|^q \right)^{1/q}$ with $\|\bu\|_0$ defined as the number of nonzero elements. For two real numbers $a$ and $b$, define $a \land b = \min(a, b)$ and $a \lor b = \max(a,b)$. We write $a \lesssim b$ if $a \le C b$ for some constant $C>0$.

Let $\{ \bX(t): t\in \tdomain\}$ be a vector-valued stochastic process defined on a compact interval $\tdomain = [0,1]$, where $\bX(t)=(X_1(t), \dots, X_p(t))^{\t}$. 
The mean and diagonal covariance functions are assumed to be continuous and denoted by $\bmu(t) = (\mu_1(t), \dots, \mu_p(t))^{\t} = E\bX(t)$ and $\Sigma(t) = (\sigma_{jk}(t))_{j,k=1}^p = E\bX(t)\bX(t)^{\t} - \bmu(t)\bmu(t)^{\t}$, respectively.
For each fixed $t$, applying multivariate PCA, we obtain
\begin{equation}\label{eq:dpca}
\bX(t) = \bmu(t) + \sum_{k=1}^p \xi_k(t) \bu_k(t), 
\end{equation}
where $\bu_k(t)$ is the $k$-th principal eigenvector and $\xi_k(t) = (\bX(t)-\bmu(t))^{\t}\bu_k(t)$ is the $k$-th principal component score with $E\xi_k(t)=0$ and $\mathrm{var}(\xi_k(t))=\lambda_k(t)$.
Without loss of generality, suppose that $\lambda_1(t) \geq \lambda_2(t) \geq \cdots \geq \lambda_p(t) \geq 0$. 
Moreover, we have $\cov(\xi_k(t), \xi_l(t))=0$ and $\bu_k(t)^{\t}\bu_l(t)=0$ for $k\neq l$ at each $t$. The time-varying version of PCA \eqref{eq:dpca} is called dynamic PCA (DPCA).

\begin{remark}\label{rmk:dpca-fpca}
	As discussed in Section \ref{sec:intro}, the DPCA is essentially different from FPCA. In particular, the DPCA represents data in the Euclidean space, i.e., it applies multivariate PCA at each $t$ to obtain $\bX(t) = \bmu(t) + \sum_{k=1}^p \xi_k(t) \bu_k(t)$, where $\bu_k(t), k=1,\dots,p$ form an orthonormal basis in $\real^p$. By comparison, the FPCA represents $\bX(t)$ in the infinite-dimensional function space, i.e., $\bX(t) = \bmu(t) + \sum_{k=1}^{\infty} \theta_k \boldsymbol{\phi}_k(t)$ where $\theta_k$ are uncorrelated functional principal scores and $\boldsymbol{\phi}_k(t)$ are orthonormal basis functions in the space of square integrable functions $L^2(\tdomain)$.  
	Note that, instead of the auto-covariance function $C(s,t)=E\bX(s)\bX(t)^{\t} - \bmu(s)\bmu(t)^{\t}$, the DPCA studies the much simper diagonal covariance function $\Sigma(t)=C(t,t)$.
\end{remark}

One advantage of DPCA is the ability to capture the dynamic information contained in data, which facilitates interpretation.
Formally, the dynamic principal eigenvectors can be found by solving the following optimization problem,
\begin{eqnarray}\label{opt:dpca}
\min_{V(t)} & \int_{\tdomain} E\|\bX(t) - \bmu(t) - V(t)V(t)^{\t}\{\bX(t)-\bmu(t)\}\|^2 dt \nonumber\\
s.t. & V(t)^{\t} V(t) = I_d,
\end{eqnarray}
where $V(t) \in \real^{p \times d} $, $I_d$ is a $d \times d$ identity matrix and $d$ is the number of principal eigenvectors of interest. 
The problem \eqref{opt:dpca} is reduced to perform multivariate PCA at each $t$ based on Lemma \ref{lem:dpca} in the Appendix, and a similar result can be found in \citet{berrendero2011principal}. 
Note that if $\lambda_d(t) - \lambda_{d+1}(t) > 0$, then $U(t) =  (\bu_1(t), \dots, \bu_d(t))$ is unique up to an orthogonal matrix, that is, $U(t)O$ is also an optimal solution of \eqref{opt:dpca} for any $d \times d$ orthogonal matrix $O$.
We refer to the subspace $\mathcal{S}(t)$ spanned by the column vectors of $U(t)$ as the dynamic principal subspace, and the corresponding projection matrix is given by $\Pi(t) = U(t)U(t)^{\t}$.

In reality, we observe noisy measurements at common or irregular design points, $y_{ijl} = x_{ij}(t_{il}) + \epsilon_{ijl}, t_{il} \in \tdomain,$
where $\epsilon_{ijl}$ are independent and identically distributed (i.i.d.) measurement errors independent of $x_{ij}$ with mean zero and variance $\sigma^2$, $i=1, \dots, n; j=1, \dots, p$ and $l = 1, \dots, m_i$, where $m_i$ is the number of observations for each trajectory of the $i$-th subject. We denote $N = \sum_{i=1}^n m_i$ and $\bar{m} = \sum_{i=1}^n m_i / n$. 
Under the common design, all the observations are sampled at the same locations, i.e., $t_{1l} = t_{2l} = \cdots = t_{nl} = t_l$ for all $l=1, \dots, m$ where $m=\bar{m}=m_1=\cdots=m_n$, while the locations $t_{il}$ are sampled independently from a compact interval $\tdomain$ under the irregular design \citep{cai2011optimal}.

An empirical version of (\ref{opt:dpca}) is formulated by substituting the expectation with its estimate. 
A naive estimate is to use the sample covariance matrix $S$, which however has some drawbacks as discussed in Section \ref{sec:intro}. First, since the data are collected at discrete grids, one can merely obtain estimates at observed times instead of the whole period $\tdomain$. Second, the sample covariance matrices are infeasible under the irregular design.
Third, the estimates may fluctuate significantly without considering smoothness. 
Therefore, a reliable and smooth estimate is desirable. 
To illustrate the main idea, we assume $\bmu(t)=0$ for the moment. 
To obtain the estimate at the target time $t$, we borrow the information of the data observed at neighboring grids. 
Thus, motivated by the local linear smoothing, we propose an empirical optimization problem as follows,
\begin{eqnarray*}\label{opt:em-dpca}
	\min_{V(t)} & & \sum_{i=1}^n \sum_{l=1}^{m_i} w_{il}(t) \|\by_{il} - V(t)V(t)^{\t}\by_{il}\|^2  \\
	s.t. & & V(t)^{\t}V(t) = I_d, \nonumber
\end{eqnarray*}
where $w_{il}(t)= \{R_2 K_h(t_{il}-t) - R_1K_h(t_{il}-t)(t_{il}-t)\}/\{R_0R_2 - R_1^2\}$, $R_{\ell} = \sum_{i=1}^n\sum_{l=1}^{m_i} K_h(t_{il}-t)(t_{il}-t)^{\ell}$, $\ell=0, 1, 2$, $h$ is the bandwidth, $K_h(\cdot) = K(\cdot/h)/h$ and $K$ is a kernel function \citep{jianqing1996local}. It can be equivalently formulated as
\begin{eqnarray}\label{opt:smoothS}
\max_{V(t)}  & & \mathrm{Tr}[ V(t)^{\t} \hat{\Sigma}(t) V(t) ]  \\
s.t. & & V(t)^{\t} V(t) = I_d, \nonumber
\end{eqnarray}
where $\hat{\Sigma}(t) = \sum_{i=1}^n \sum_{l=1}^{m_i} w_{il}(t) \by_{il}\by_{il}^{\t}$ is the smoothed covariance matrix. 

Note that our proposal readily adapts to both common and irregular designs using pooled data. More generally, incorporating the estimated mean function by the kernel smoothing, the estimator $\hat{U}(t)$ can be obtained by substituting $\hat{\Sigma}(t)$ in (\ref{opt:smoothS}) with
\begin{eqnarray*} \label{eq:smoothS?general}
	\hat{\Sigma}(t) & = & \sum_{i=1}^n \sum_{l=1}^{m_i}  w_{il}(t)\by_{il}\by_{il}^{\t} -  \sum_{i=1}^n \sum_{l=1}^{m_i} w_{il}(t) \by_{il} \sum_{i=1}^n \sum_{l=1}^{m_i} w_{il}(t) \by_{il}^{\t}.
\end{eqnarray*}
In addition, under the common design where the data are observed at regular grids, practitioners can adopt an alternative estimate of $\hat{\Sigma}(t)$ for simplified computation,
\begin{equation}\label{eq:smoothS-common}
\hat{\Sigma}_{common}(t) =  \sum_{l=1}^{m} w_{l}(t) \sum_{i=1}^n n^{-1} (\by_{il}- \bar{y}_l) (\by_{il} - \bar{y}_l)^{\t}, 
\end{equation}
where $\bar{y}_l = \sum_{i=1}^n y_{il}/n$, $w_l(t) = \left\{ R_{2,c}K_h(t_l - t) - R_{1,c}K_h(t_l - t)(t_l-t) \right\} / (R_{2,c}R_{0,c} - R_{1,c}^2)$, $R_{\ell,c} = \sum_{l=1}^m K_h(t_l -t)(t_l -t)^{\ell}$, $\ell = 0, 1, 2$.

\subsection{Sparsity and estimation in high dimensions} \label{subsec:sparsity}

For high-dimensional data, the number of variables $p$ is comparable to or even much larger than the sample size $n$. The estimator $\hat{U}(t)$ in Section \ref{subsec:dps} may become inconsistent without additional structures. The sparsity assumption is necessary to enhance the interpretability and improve the estimates in high dimensions. 
Assume that $U(t) \in \mathcal{U}(q, R_q; \tdomain)$, where
\begin{equation*}\label{eq:lqset}
\mathcal{U}(q, R_q; \tdomain) = \left\{ U(t) \in \real^{p \times d}, t \in \tdomain \bigg| U(t) \in \mathbb{V}_{p,d}, \sup_{t \in \tdomain}\max_{1 \leq j \leq d} \|\bu_j(t)\|_q^q \leq R_q \right\},
\end{equation*}
with $0 < q \leq 1$. When $q=0$,
\begin{equation*}\label{eq:l0set}
\mathcal{U}(0, R_0; \tdomain) = \left\{ U(t) \in \real^{p \times d}, t \in \tdomain \bigg| U(t) \in \mathbb{V}_{p,d}, \sup_{t \in \tdomain}\max_{1 \leq j \leq d} \|\bu_j(t)\|_0 \leq R_0 \right\}.
\end{equation*}
The set $\mathcal{U}(q, R_q; \tdomain)$ is non-empty and the $l_q$ constraint is active only when $1 \leq R_q \leq p^{1-q/2}$. In Section \ref{sec:thm}, we consider bounded $R_q$ to simplify the theoretical exposition.
The family of leading eigenvectors over $\tdomain$ defined in $\mathcal{U}(q, R_q; \tdomain)$ generalizes the notion of static eigenvectors in \citet{vu2012minimax}. We stress that if the sparsity condition does not hold uniformly over $\tdomain$, our method can still apply to the subregions of $\tdomain$ where this condition holds. 

Recall that the projection matrix is $\Pi(t) = U(t)U(t)^{\t}$. By definition, $\Pi_{jj}(t)=0$ holds if and only if each element of the $j$-th row of $U(t)$ is zero. Further, it implies that if $\Pi_{jj}(t)=0$, then all entries of the $j$-th row/column of $\Pi(t)$ are 0. Denote the support set $J(t) = \{j: \Pi_{jj}(t)>0\}$.
For notational convenience, we introduce the following block representation of $\Sigma(t)$:
\[ \left( \begin{array}{cc}
\Sigma_{JJ}(t) & \Sigma_{JJ^c}(t) \\
\Sigma_{J^cJ}(t) & \Sigma_{J^cJ^c}(t)
\end{array}
\right). \]
Similar block representations can be defined for other matrices and vectors.
Apparently, the {principal eigenvectors} at $t$ does not depend on the variables outside of the set $J(t)$ in the sense that all elements of $U_{J^c}(t)$ are 0.

Note that $U(t) \in \mathbb{V}_{p,d}$, where $\mathbb{V}_{p,d} = \{ V \in \real^{p \times d}| V^{\t}V = I_d \}$ is the Stiefel manifold, which results in a non-convex problem which is hard to solve. 
Most existing algorithms for static sparse PCA require deflation steps \citep{shen2008sparse,mackey2009deflation} or convex relaxation \citep{d2005direct,vu2013fantope} to circumvent the orthogonality constraint, which either lack theoretical guarantees or are computational expensive. 
To avoid these issues, our optimization problem is defined directly on the Stiefel manifold $\mathbb V_{p,d}$ which can be solved by manifold optimization algorithms.
The regularized manifold optimization problem is formulated as follows,
\begin{equation}\label{opt:sparse}
\begin{array}{cc}
\min \limits_{V(t)} & -\mathrm{Tr}[ V(t)^{\t} \hat{\Sigma}(t) V(t) ] + \rho_t \|V(t)\|_1 \\
s.t. & V(t)^{\t} V(t) = I_d,
\end{array}
\end{equation}
where $\rho_t>0$ is the regularization parameter at $t$. We allow the parameter $\rho$ to depend on $t$, which makes our proposal fully adaptive to different sparsity levels varying with $t$.
The optimization problem \eqref{opt:sparse} deals with sparsity and orthogonality jointly, which can be solved effectively with recent developments for manifold optimization, e.g., the proximal gradient method \citep{chen2020proximal}. 

To improve estimation, we treat the solution of \eqref{opt:sparse} as an initial estimate which is denoted by $\hat{U}^0(t)$, and then propose a refined version. Specifically, we add a thresholding step to further filter out the variables irrelevant to the principal eigenvectors. Denote the set of remaining variables by $\hat{J}(t) = \{j: \hat{\Pi}_{jj}^0(t) \geq \gamma_t\}$, where $\hat{\Pi}^0(t)=\hat{U}^0(t)\hat{U}^0(t)^{\t}$ and $\gamma_t > 0$ is the thresholding parameter at $t$. Since the estimate after thresholding may not belong to the Stiefel manifold $\mathbb{V}_{p,d}$, we re-estimate the principal eigenvectors afterwards. Our refined estimate is given by
\[\hat{U}(t) = \left(\begin{array}{c}
\hat{U}_{\hat{J}(t)}(t) \\
0
\end{array} \right),\]
where $\hat{U}_{\hat{J}(t)}(t)$ is the solution of the problem,
\begin{equation}\label{pro:refineop}
\begin{array}{cc}
\min \limits_{V(t)} & -\mathrm{Tr}\{ V(t)^{\t} \hat{\Sigma}_{\hat{J}(t)\hat{J}(t)}(t) V(t) \} + \rho_t \|V(t)\|_{1,1} \\
s.t. & V(t)^{\t} V(t) = I_d.
\end{array}
\end{equation}

Note that the estimated principal subspaces are readily obtained by spanning the columns of $\hat{U}(t)$ with projection matrices $\hat{\Pi}(t) = \hat{U}(t)\hat{U}(t)^{\t}$.
The two-step estimation procedure successfully identifies the significant variables and provides consistent estimators under general $l_q$ sparsity, which is theoretically and empirically demonstrated in Sections \ref{sec:thm} and \ref{sec:sim}.

\subsection{Tuning parameters}\label{subsec:tuning}
In this section, we discuss how to select parameters that are involved in the estimation procedure. 
Note that in the dynamic setting, the number of principal eigenvectors of interest $d$ may be a constant or vary with $t$. There exists no consensus on the selection of $d$ which depends on the specific application. For example, it could be selected based on the fraction of variance explained (FVE). 
In supervised problems such as regression or classification, it may be tuned by $k$-fold cross validation to minimize the prediction/classification error. 
Here we mainly consider tuning three other parameters, the bandwidth $h$, the sparsity parameter $\rho_t$ and the thresholding parameter $\gamma_t$. We suggest to select them in a sequential manner \citep{chen2015localized,chen2016dynamic}. For the bandwidth $h$, we use the leave-one-curve-out cross-validation approach \citep{rice1991estimating,yao2005functionala}.
Specifically, we tune the bandwidth $h$ given $\rho_t=0$ and $\gamma_t=0$ by maximizing the cross-validated inner product,
\begin{equation*}
h^{*} = \mathop{\arg \max}_{h \in \mathcal{A}_1} \frac{1}{n\bar{m}} \sum_{i=1}^n \sum_{l=1}^{m_i} \mathrm{Tr}\{ \hat{U}_{h,0,0}^{-i}(t_{il})^{\t}(\by_{il} - \hat{\bmu}(t_{il})) (\by_{il}-\hat{\bmu}(t_{il}))^{\t} \hat{U}_{h,0,0}^{-i}(t_{il}) \},
\end{equation*}
where $\hat{\bmu}$ is the estimated mean function which refers to the sample mean under the common case and the local linear estimate under the irregular design, $\mathcal{A}_1$ is a candidate set of $h$, $\hat{U}_{h,0,0}^{-i}$ is estimated by leaving out the $i$-th subject with the bandwidth $h$, $\rho_t=0$ and $\gamma_t=0$. 
Next, the parameter $\rho_t$ is determined by $k$-fold cross-validation. The data is divided into $k$-folds by subjects, denoted by $\mathcal{D}_1, \dots, \mathcal{D}_k$. Let $\hat{U}_{h, \rho_t, \gamma_t}^{-\nu}(t)$ be the estimator using data other than $\mathcal{D}_{\nu}$ at time $t$ with parameters $h$, $\rho_t$ and $\gamma_t$. Let $\hat{\Sigma}_h^{\nu}(t)$ be the smoothed covariance matrix estimate at $t$ using $\mathcal D_{\nu}$ with the bandwidth $h$. Next, we choose $\rho_t$ given the selected bandwidth $h^{*}$ and $\gamma_t=0$ by maximizing the cross-validated inner product,
\begin{equation*}\label{tune:rhocv}
\rho_t^{*} = \mathop{\arg\max}_{\rho_t \in \mathcal{A}_{2,t}} \frac{1}{k} \sum_{\nu=1}^{k} \mathrm{Tr} [ \{\hat{U}_{h^{*}, \rho_t, 0}^{-\nu}(t)\}^{\t} \hat{\Sigma}_{h^{*}}^{\nu}(t) \hat{U}_{h^{*}, \rho_t, 0}^{-\nu}(t)],
\end{equation*}
where $\mathcal{A}_{2,t}$ is a candidate set for $\rho_t$.
At last, we tune the thresholding parameter $\gamma_t$, given the selected bandwidth $h^{*}$ and sparsity level $\rho_t^{*}$, by a trade-off between the explained variance $Ip(\gamma_t)$ and model complexity, i.e., the number of retained variables, where
\begin{equation*}\label{tune:gammacv}
Ip(\gamma_t) =  \frac{1}{k} \sum_{\nu=1}^{k} \mathrm{Tr} [ \{\hat{U}_{h^{*}, \rho_t^{*}, \gamma_t}^{-\nu}(t)\}^{\t} \hat{\Sigma}_{h^{*}}^{\nu}(t) \hat{U}_{h^{*}, \rho_t^{*}, \gamma_t}^{-\nu}(t)],
\end{equation*}
where $\gamma_t \in \mathcal{A}_{3,t}$, $\mathcal{A}_{3,t}$ is a candidate set and $Ip(\gamma_t)$ is the cross-validated inner product when the threshold equals $\gamma_t$. The model complexity depicts the cardinality of the support set $\hat{J}(t)$. One can select the $\gamma_t$ to achieve model parsimony without much information loss. We demonstrate the performance of selected parameters in Section \ref{sec:sim}.

\section{Theoretical results}\label{sec:thm}

In this section, we investigate the theoretical properties of the proposed estimator under both common and irregular designs.
To measure the performance of the estimator, we use the notion of the distance defined in \citet{vu2013minimax}.
For $U, V \in \mathbb{V}_{p,d}$, the squared distance is defined by
\begin{equation}\label{eq:distance}
d^2(U, V) = d^2(\mathcal{E}, \mathcal{F}) = \frac{1}{2}\|E - F\|_F^2, 
\end{equation}
where $\mathcal{E}$ and $\mathcal{F}$ are subspaces with projection matrices $E=UU^{\t}$ and $F=VV^{\t}$, respectively.

Some assumptions necessary for theoretical results are provided, concerning the properties of variables and kernel functions. Assumption \ref{assump:eigengap} ensures that the $d$-dimensional principal subspace is well-defined.
In multivariate cases, the commonly used assumption for sparse PCA is that $x_j^2$ is sub-exponential, while Assumption \ref{assump:tail} is adapted to random processes, which holds rather generally, e.g., Gaussian processes.

\begin{assumption}\label{assump:eigengap}
	Assume that $\lambda_d(t) - \lambda_{d+1}(t) > 0$ for all $t \in \tdomain$.
\end{assumption}

\begin{assumption}\label{assump:tail}
	For each $j=1, \dots, p, X_j^2(t)$ is sub-exponential uniformly in $t \in \tdomain$, that is, there exists a positive constant $\lambda_0$ such that $\sup_{t \in \tdomain} Ee^{\lambda X_j^2(t)} < \infty$ for $|\lambda| < \lambda_0$. Also assume the measurement error $\epsilon^2$ is sub-exponential.
\end{assumption}

\begin{assumption}\label{assump:smoothmean-cov}
	The mean functions $\mu_j(\cdot)$ and the diagonal covariance functions $\sigma_{jk}(t)$ are twice differentiable and the second derivative is bounded on $\tdomain$ for $j,k=1, \dots, p$.
\end{assumption}

\begin{assumption}\label{assump:p}
	Assume $\log p(n\bar{m}^{-1}+ nh)^{-1} \to 0$ as $n \to \infty$.
\end{assumption}

\begin{assumption}\label{assump:kernel}
	The kernel function $K(\cdot)$ is a bounded and symmetric probability density function on $[-1,1]$ with $\int u^2 K(u) du < \infty$ and $\int K^2(u) du < \infty$. 
\end{assumption}

The smoothness of mean and diagonal covariance functions is imposed in Assumption \ref{assump:smoothmean-cov}, while $\bX(t)$ is not necessarily smooth. 
Assumption \ref{assump:p} indicates $\log p = O(n^c)$ for some $c>0$ since $\bar{m}$ and $h^{-1}$ typically grow at a fractional polynomial order of $n$.
Assumption \ref{assump:kernel} is standard in the kernel smoothing literature \citep{jianqing1996local,chen2016dynamic}.

First, we quantify the performance of the thresholding step by investigating the false positive control and false negative control of $J(t)$. It is revealed in Lemma \ref{lemma:support} that, with a suitable parameter $\gamma_t$, we can recover the support set consistently.
The condition $\min_{j \in J(t)} \Pi_{jj}(t) \geq 2\gamma_t$ assures that the important variables can be distinguished from the noise stochastically. Denote $\Gamma(t) = \Sigma(t) + \sigma^2 I_p$ where $I_p \in \real^{p \times p}$ is an identity matrix.

\begin{lemma}\label{lemma:support}
	Assume $U(t) \in \mathcal{U}(q, R_q; \tdomain)$ and recall that $\hat{\Pi}^0(t)$ is the initial estimator. Note that we have $\|\Pi(t) - \hat{\Pi}^0(t)\|_F^2 \leq C\|\hat{\Sigma}(t) - \Gamma(t)\|_{\infty} = o_p(1)$, $t \in \tdomain$ for some positive constant $C>0$. If $\min_{j \in J(t)} \Pi_{jj}(t) \geq 2\gamma_t$ and $\gamma_t > \|\Pi(t) - \hat{\Pi}^0(t)\|_F$, then the variable selection procedure $\hat{J}(t) := \{ j: \hat{\Pi}_{jj}^0(t) \geq \gamma_t \}$ succeeds.
\end{lemma}

In the following, we state the theoretical properties of the eventually obtained estimators $\hat{U}(t)$.
A theoretical challenge is how to carefully control the $l_q (0 \le q \le 1)$ norm of our estimators obtained with a lasso-type penalty, which can be tackled by the consistent variable selection. 
Moreover, we need to deal with the dependence between observations from the same trajectory with care to control the concentration bound of the local linear estimator, which is essential to the theoretical results.
Next, we investigate the behavior of the resulting estimator under both irregular and common designs.

\subsection{Rate of convergence under the irregular design}\label{subsec:thm-indpt}

In this section, we provide a theoretical investigation of estimators under the irregular design. 
Assumption \ref{assump:indpt-sampling} is about the sampling scheme under the irregular design \citep{cai2011optimal}.
Assumption \ref{assump:indpt-m} on the sampling frequency is assumed to quantify the within-subject dependence and facilitate the exposition of theoretical analysis. This is a standard condition used in functional data \citep{zhang2016sparse}, which is fairly mild and holds for the common design and the irregular design with finite $m_i$ or when $m_i$ are not all vastly different, $1\le i \le n$.

\begin{assumption}\label{assump:indpt-sampling}
	Under the irregular design, $t_{il}, i=1,\dots,n; l=1,\dots,m_i$ are independent and identically distributed from a density $f(\cdot)$ with compact support $\tdomain$. In addition, the sampling density $f_{\tdomain}$ is bounded away form zero and infinity and is twice continuously differentiable with a bound derivative on its support. 
\end{assumption}

\begin{assumption}\label{assump:indpt-m}
	Assume $\lim \sup_n \sum_{i=1}^n m_i^2 / n\bar{m}^2 < \infty$ and $\sup_n (n\max_i m_i/\sum_{i=1}^n m_i) < \infty$.
\end{assumption}

\begin{theorem}\label{thm:indpt}
	Suppose that $U(t) \in \mathcal{U}(q, R_q; \tdomain)$ for $0 \le q \le 1$. Under Assumptions \ref{assump:eigengap}-\ref{assump:indpt-m}, for a fixed point $t\in [0,1]$, if $\rho_t = O\left[\{\log p/(n\bar{m}h) + \log p/n\}^{1/2} + h^2\right]$ and $\min_{j \in J(t)} \Pi_{jj}(t) \geq 2\gamma_t$, where $ \gamma_t^2 = O\left[\{\log p/(n\bar{m}h) + \log p/n\}^{1/2} + h^2\right]$, then
	\begin{equation*}
	d\{U(t), \hat{U}(t)\} = O_p\left[ \left\{\left(\frac{\log p}{n\bar{m}h}+ \frac{\log p}{n}\right)^{1/2} + h^2 \right\}^{1-q/2} \right].
	\end{equation*}
\end{theorem}

From Lemma \ref{lemma:support}, the condition $\min_{j \in J(t)} \Pi_{jj}(t) \geq 2\gamma_t$ together with the choice of $\gamma(t)$ in Theorem \ref{thm:indpt} ensures that signal variables can be distinguished from the noise, which leads to consistent variable selection. The parameter $\rho_t$ is to balance the trade-off between bias and variance.
The rate of convergence in Theorem \ref{thm:indpt} consists of two parts, the variance term $\left\{\log p/(n\bar{m}h)+ \log p/n\right\}^{1/2}$ and the bias term $h^2$ for $q=0$, which is consistent with that of the mean estimation in \citet{zhang2016sparse} up to the $\log p$ term accounting for high dimensionality. The convergence rate depends on $\bar{m}$ through the total number of observations $n\bar{m}$. Thus, the magnitude of $\bar{m}$ can be of any order of the sample size $n$ as long as $\log p/(nh) \to 0$ and $h \to 0$, which demonstrates the advantage of our proposal in handling the sparsely observed data.
A careful inspection shows that the convergence rate exhibits a phase transition phenomenon. When $\bar{m}h \to \infty$, the sampling frequency $\bar{m}$ has no effect on the resulting rate, $(\log p /n)^{1/2-q/4}$, as if the whole curves are completely observed. Otherwise, the estimates attain the nonparametric rate $\left[\left\{\log p/(n\bar{m}h)\right\}^{1/2} + h^2 \right]^{1-q/2}$ as if all $n\bar{m}$ observations are independently observed.

In contrast, the convergence rate is of the order $[\{\log p/(nh)\}^{1/2} + h^2 ]^{1-q}$ for the dynamic covariance estimation in \citet{chen2016dynamic} under the assumption that the columns of covariance matrices possess the $l_q$-type sparsity structure. There are two notable differences between the two rates. First, in our setting, the effective sample size is $(n\bar{m}h) \land n$ instead of $nh$ which differs from the conventional nonparametric scheme. This is because we take the correlation among observations from the same subject into account which is an important nature of functional data or generally the repeated measurements data.
Second, the dependence on the $q$ for these two convergence rates is different. 
In the static case, it is known that the optimal rate for eigenvector estimation $(\log p/n)^{1/2-q/4}$ is faster than the rate obtained for covariance estimation $(\log p/n)^{1/2-q/2}$ \citep{bickel2008covariance,vu2012minimax,cai2012minimax}. Likewise in the dynamic setting, the convergence rate for eigenvectors in Theorem \ref{thm:indpt} is faster than the rate of the corresponding covariance estimation.
The theoretical finding is corroborated in empirical studies that the eigenvector estimators based on the eigen-decomposition of dynamic covariance estimates perform sub-optimally.

Moreover, Lemma \ref{lemma:error_varinfty} reveals that $d\{U, \hat{U}\} \leq C_q\|\hat{\Sigma} - \Gamma\|_{\infty}^{1-q/2}$ where $C_q = Cd^2R_q$ for some positive constant $C>0$. Therefore, the quantity $R_0$ may be allowed to grow to infinity, and the consistency of the estimator is guaranteed as long as $R_0 \{\left(\log p/(nmh)+ \log p/n \right)^{1/2} + h^2 \} \to 0$. 
Given that the bandwidth is carefully tuned to balance the bias and variance, the quantity $\bar{m}$ plays a crucial role in the convergence rate, which is illustrated in Corollary \ref{corollary:indpt}.

\begin{corollary}\label{corollary:indpt}
	
	Suppose that $U(t) \in \mathcal{U}(q, R_q)$ and conditions in Theorem \ref{thm:indpt} hold and $t$ is a fixed point in $[0,1]$.
	
	(1) When $\bar{m}/(n/\log p)^{1/4} \to 0$ and $h = O[\{\log p/ (n\bar{m})\}^{1/5}]$,
	\[ d\{U(t), \hat{U}(t)\} = O_p\left[ \left\{\left(\frac{\log p}{n\bar{m}h}\right)^{1/2} + h^2 \right\}^{1-q/2} \right]. \]
	
	(2) When $\bar{m}/(n/\log p)^{1/4} \to C$, where $C >0$, and $h = O\{(\log p/ n)^{1/4}\}$,
	\[ d\{U(t), \hat{U}(t)\} = O_p\left\{ \left(\frac{\log p}{n}\right)^{1/2 - q/4}\right\} . \]
	
	(3) When $\bar{m}/(n/\log p)^{1/4} \to \infty$, $h = o\{(\log p/ n)^{1/4}\}$ and $\bar{m}h \to \infty$,
	\[ d\{U(t), \hat{U}(t)\} =  O_p\left\{ \left(\frac{\log p}{n}\right)^{1/2 - q/4}\right\}.\]
\end{corollary} 

As Corollary \ref{corollary:indpt} reveals, the phase transition occurs when $\bar{m}$ is of the order $(n/\log p)^{1/4}$.
When $\bar{m}$ is relatively small as in case (1), the nonparametric rate is determined jointly by quantities $n$ and $\bar{m}$. With $\bar{m}$ grows such that $\bar{m} \gtrsim (n/\log p)^{1/4}$, the rate achieves $(\log p/n)^{1/2-q/4}$ regardless of $\bar{m}$ which coincides with the optimal rate for estimating static eigenvectors. Although the rates are of the same order in cases (2) and (3) which fall into the parametric paradigm, the bias in case (2) is non-vanishing \citep{zhang2016sparse}.  
With the advantage of data pooling, the grids are allowed to be sparse under the irregular design as long as the sample size $n$ suffices. 

\subsection{Rate of convergence under the common design}\label{subsec:thm-common}

In this section, we focus on the common design where sampling locations $t_l, l=1, \dots, m$ are deterministic. 

\begin{assumption}\label{assump:common-sampling}
	Under the common design, $t_l$'s are fixed and distinct, and $\max_{0 \leq l \leq m} |t_{l+1} - t_l| \leq C m^{-1}$, where $t_0=0, t_{m+1}=1$.
\end{assumption}

\begin{assumption}\label{assump:mcommon}
	The sampling frequency $m \to \infty$ and $1/(mh) = O(1)$, $h\to 0$ as $n \to \infty$.
\end{assumption}

Under the common design, the data should be observed on sufficiently dense grids. To see this, if $m$ is finite, no data is available in the suitably small neighboring region for some $t$, which causes large bias for the resulting estimates.
Moreover, Assumption \ref{assump:mcommon} guarantees that $h \ge \min_{j=1,\dots, m} |t - t_j| = O(1/m)$ for each $t \in \tdomain$ to avoid the trivial estimator.

\begin{theorem}\label{thm:common} Suppose that $U(t) \in \mathcal{U}(q, R_q; \tdomain)$ for $0 \le q \le 1$. Under Assumptions \ref{assump:eigengap}-\ref{assump:kernel},\ref{assump:common-sampling} and \ref{assump:mcommon}, for a fixed point $t \in [0,1]$, if $\rho_t = O\left[\{\log p/(nmh) + \log p/n\}^{1/2} + h^2 \right]$ and $\min_{j \in J(t)} \Pi_{jj}(t) \geq 2\gamma_t$, where $ \gamma_t^2 = O\left[\{\log p/(nmh) + \log p/n\}^{1/2} + h^2 \right]$, then
	\begin{equation*}
	d\{U(t), \hat{U}(t)\} = O_p\left[ \left\{\left(\frac{\log p}{nmh}+ \frac{\log p}{n}\right)^{1/2} + h^2 \right\}^{1-q/2} \right].
	\end{equation*}
\end{theorem}

At first glance, the convergence rates under common and irregular designs are similar. However, since the data are observed at common locations under this design, the number of locations where the data are used for estimation is of the order $mh$. Consequently, the sampling frequency $m$ is required to be sufficiently large which differs from the case under the irregular design. The effect of $m$ on the convergence rate is illustrated in Corollary \ref{corollary:common}.

\begin{corollary}\label{corollary:common}
	Suppose that $U(t) \in \mathcal{U}(q, R_q)$ and conditions in Theorem \ref{thm:common} hold and $t$ is a fixed point in $[0,1]$.
	
	(1) When $m/(n/\log p)^{1/4} \to 0$ and $h = O(1/m)$,
	\[ d\{U(t), \hat{U}(t)\} = O_p\left\{ \left(\frac{1}{m^2}\right)^{1-q/2} \right\} . \]
	
	(2) When $m/(n/\log p)^{1/4} \to C$, where $C >0$, and $h = O\{(\log p/ n)^{1/4}\} = O(1/m)$,
	\[ d\{U(t), \hat{U}(t)\} =  O_p\left\{ \left(\frac{\log p}{n}\right)^{1/2 - q/4} \right\}. \]
	
	(3) When $m/(n/\log p)^{1/4} \to \infty$, $h = o\{(\log p/ n)^{1/4}\}$ and $mh \to \infty$, 
	\[ d\{U(t), \hat{U}(t)\} = O_p\left\{ \left(\frac{\log p}{n}\right)^{1/2 - q/4} \right\}.\]
\end{corollary}

The phase transition phenomenon appears more complex under the common design because of the interplay among the quantities $h$, $n$ and $m$. The bandwidth is usually chosen to be of the order $\{\log p/ (n\bar{m})\}^{1/5}$ to balance the associated bias and variance. Yet, restricted by the condition $1/(mh)=O(1)$ that guarantees observations available in the local window, the bandwidth is at least of the order $1/m$, see Corollary \ref{corollary:common}. When $m$ is relatively small, the sampling frequency $m$ plays a dominant role in the convergence rate, i.e., $(1/m^2)^{1-q/2}$, which is slower than that under the irregular design. When $m$ grows such that $m \gtrsim (n/\log p)^{1/4}$, the parametric rate $(\log p/n)^{1/2-q/4}$ can be achieved. 
Again, it requires $m \to \infty$ to guarantee the consistency, otherwise, the bias is not negligible.

In summary, the phase transition occurs at the same order, $m=O\{(n/\log p)^{1/4}\}$, for both designs.
When $m/(n/\log p)^{1/4} \to 0$, the rate depends on the total number of observations $n\bar{m}$ under the irregular case, while the rate is solely determined by the sampling frequency $m$ under the common case. Otherwise, both designs achieve parametric rates.
Under the irregular design, the estimator is consistent as long as $(\log p/{n\bar{m}h}) \lor (\log p/n) \to 0$. Thus, as the sample size permits, we can handle the extremely sparse case, that is, the sampling frequency is allowed to be very small. However, the quantity $m$ should be sufficiently large to achieve reasonable estimates under the common design. Further, when the sampling frequency is small, the irregular design is preferable to the common design with faster convergence rates.

\section{Simulation}\label{sec:sim}

In this section, several experiments are conducted to evaluate the numerical performance of our proposal under irregular and common designs. The observations are generated from the model,
$
\by_i(t_{il}) = \sum_{k=1}^{10} \xi_{ik}(t_{il}) \bu_{k}(t_{il}) + \boldsymbol{\epsilon}_{il}, i=1, \dots, n; l= 1, \dots, m_i,
$
where $\xi_{ik} \in \real, \by_i(t_{il}), \bu_{k}(t_{il}) \in \real^p$ and $\boldsymbol{\epsilon}_{il} \stackrel{iid}{\sim} N_p(0, \sigma^2I_p)$. Moreover, we set $\xi_{ik}(t_{il}) \equiv \xi_{ik} \stackrel{iid}{\sim} N(0, \lambda_k)$ where $\boldsymbol{\lambda} = (30,18,10,5,3,2,1,0.5,0.2,0.1)^{\t}$, which follows the classical generation mechanism of functional data.  
The sparse eigenvectors $\bu_{k}$ are obtained by applying Gram-Schmidt orthonormalization on $\mathbf{v}_{k}(t)$ defined later. Let $ v_{k,(k-1)\times 5+r}(t) = \phi_{r}(t), k=1,\dots,10, r=1,\dots,5$, and other entries of $\mathbf{v}_k$ be zero, where $v_{k,j}$ is the $j$-th element of the vector $\mathbf{v}_k$, $\phi_r(t)$ are functions in the Fourier basis, $\phi_r(t) = \sqrt{2} \sin(\pi (r+1) t)$ when $r$ is odd, $\phi_r(t) = \sqrt{2} \cos(\pi r t)$ when $r$ is even.
The locations of non-zero elements are different for different eigenvectors.

\begin{figure}[htbp]
	\centering
	\newcommand{\thisgap}{1.5mm}
	\newcommand{\thiswidth}{0.3\linewidth}
	\begin{tabular}{ccc}
		\hspace{\thisgap}\includegraphics[width=\thiswidth]{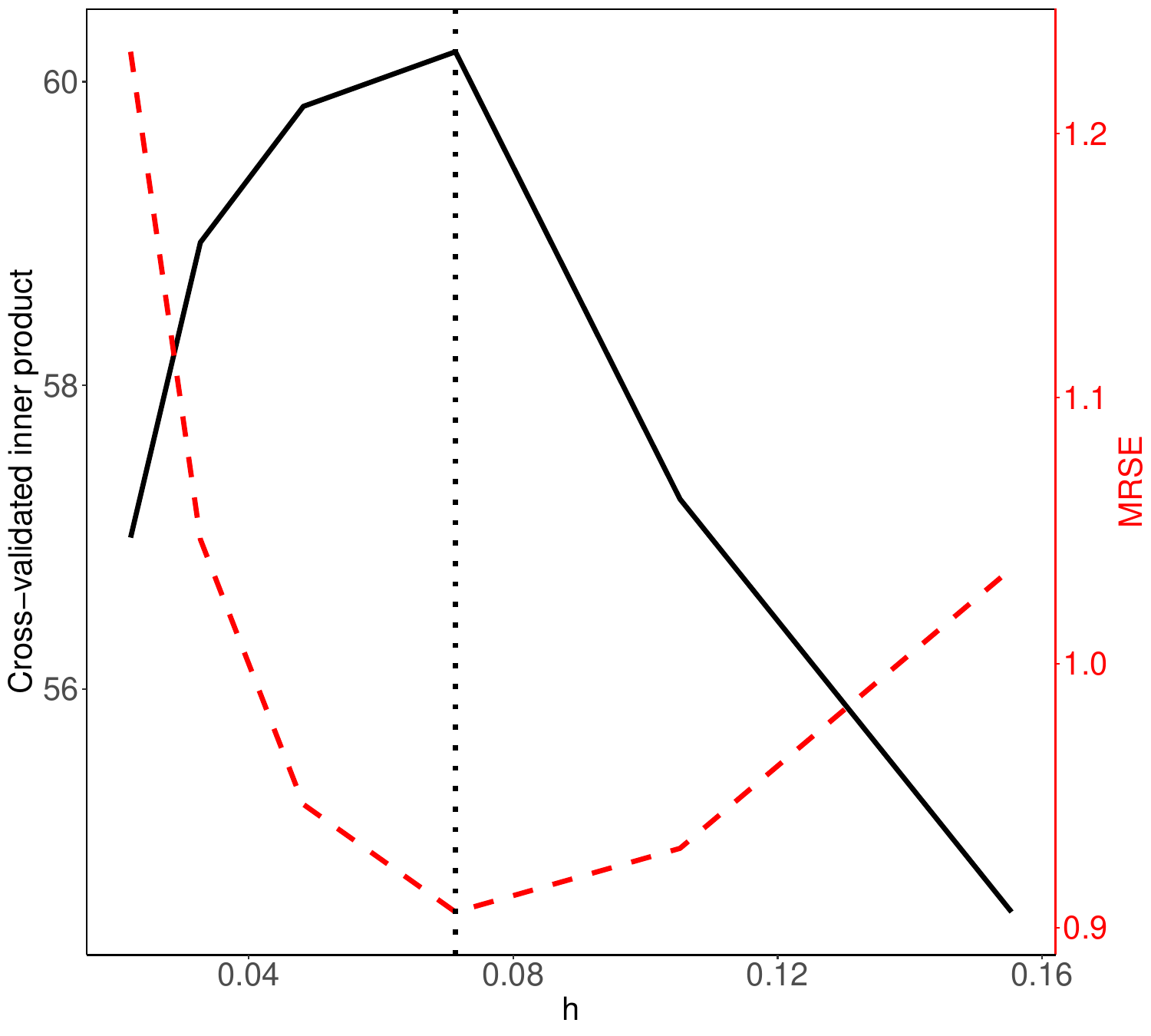} &
		\hspace{\thisgap}\includegraphics[width=\thiswidth]{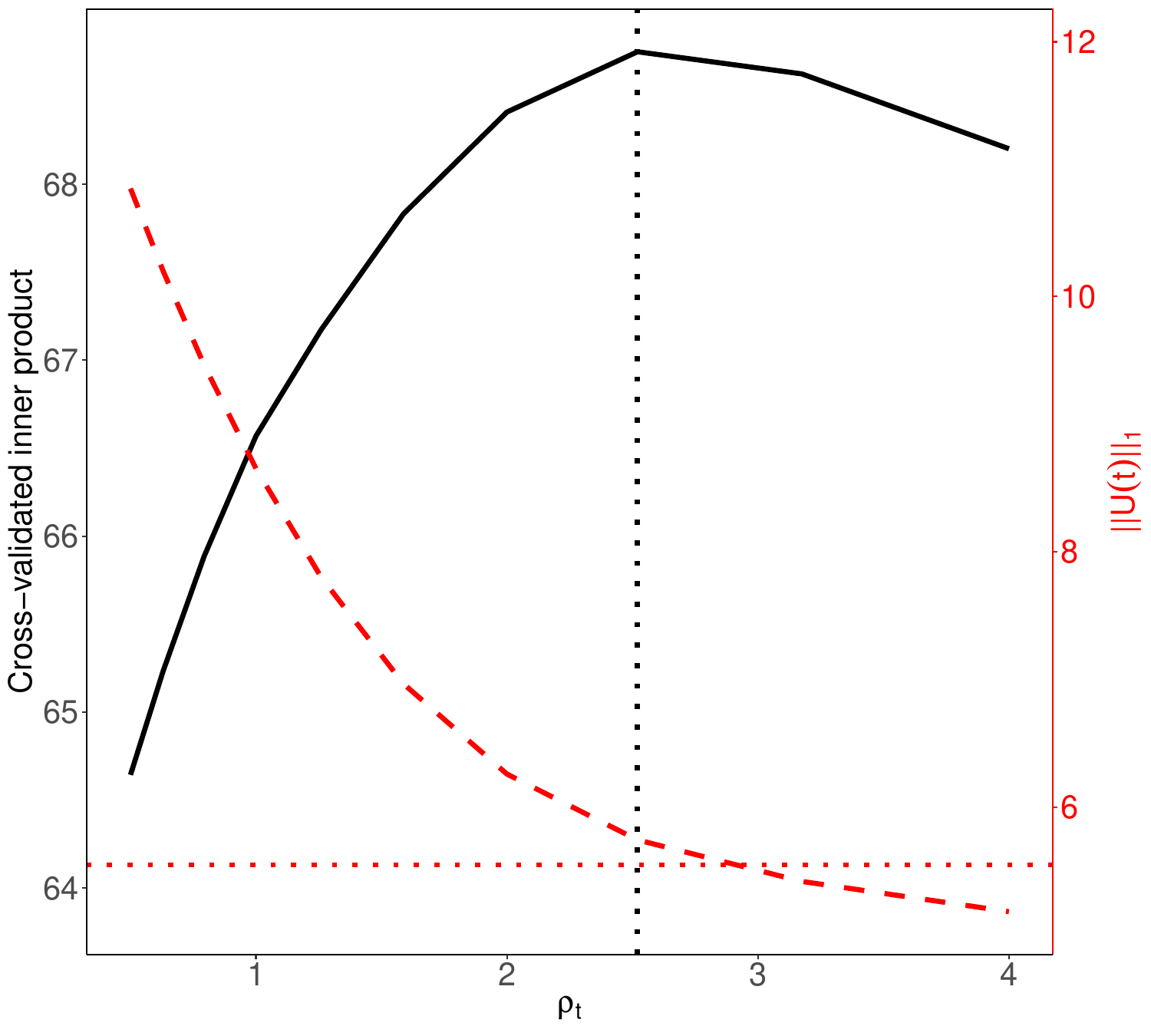} &
		\hspace{\thisgap}\includegraphics[width=\thiswidth]{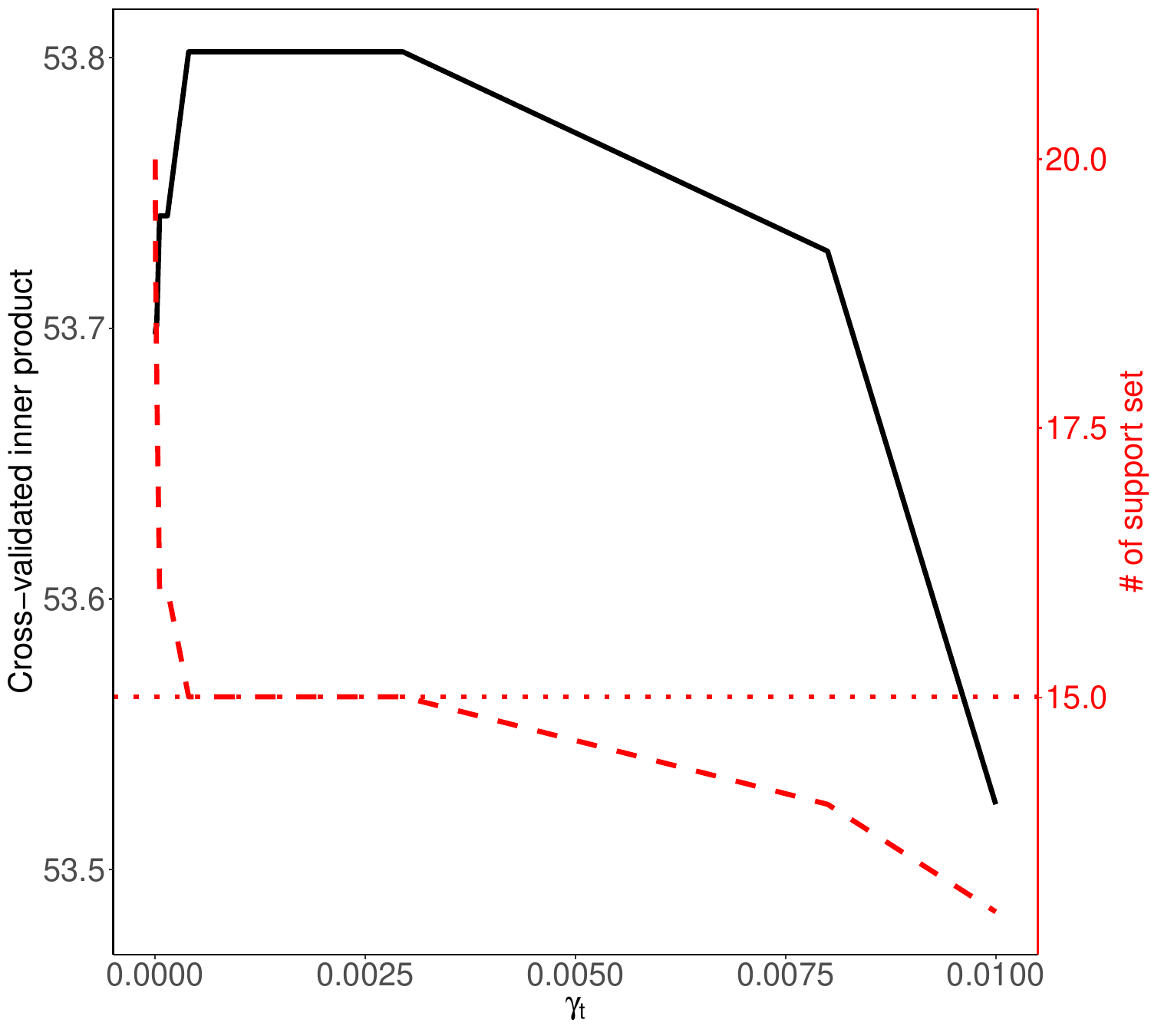} \\
	\end{tabular}
	\vspace{-2.5mm}
	\caption{The performance of the cross-validation to select parameters under the common design with $p=100, m=50$ and $\sigma^2=3$. All the parameters are selected by maximizing the cross-validated inner product (black, solid, left $y$ label). Left: bandwidth selection. The right $y$ label indicates the $MRSE(h) = m^{-1}\sum_{l=1}^m \|\hat{\Sigma}(t_l) - \Sigma(t_l)\|_F^2/\|\Sigma(t_l)\|_F^2$. An ideal $h$ should be close to the $h_{ora} = \arg\min MRSE(h)$. The selected bandwidth matches the one attaining the minimal MRSE (treated as a benchmark; see red, dashed line, right $y$ label) . Middle: sparsity parameter selection. For a fixed $t$, the right $y$ label represents $\|U(t)\|_1$. The $\|\cdot\|_1$ of estimated $U(t)$ with the selected $\rho_t$ (red, dashed, right $y$ label) approximately meets that of the true matrix (indicated by the horizontal dotted line). Right: thresholding parameter selection. The right $y$ label represents the size of the support set. The $\gamma_t$ is selected as the minimum value which maximizes the cross-validated inner product. The number of retained variables with the selected $\gamma_t$ (red, dashed, right $y$ label) is equal to the true number of relevant variables (horizontal, dotted).}
	\label{fig:cv}
\end{figure}

We design simulation settings to demonstrate the effect of the sample size $n$ and the sampling frequency under both common and irregular cases. 
Under the irregular design, we consider six settings for various combinations of $m_i$ and $n$ where $m_i$ are i.i.d from a discrete uniform distribution on the set $\mathcal M$. {\em Setting 1:} $n=100$ and $\bar{m}=100$, $\mathcal M = \{95, 100, 105\}$. {\em Setting 2:} $n=100$ and $\bar{m}=50$, $\mathcal M =\{45, 50, 55\}$. {\em Setting 3:} $n=100$ and $\bar{m}=20$, $\mathcal M = \{15, 20, 25\}$. {\em Setting 4:} $n=500$ and $\bar{m}=20$, $\mathcal M = \{19, 20, 21\}$. {\em Setting 5:} $n=500$ and $\bar{m}=10$, $\mathcal M = \{9, 10, 11\}$. {\em Setting 6:} $n=500$ and $\bar{m}=4$, $\mathcal M = \{3, 4, 5\}$.
The time points $t_{il}$ are i.i.d. sampled from the uniform distribution on [0,1]. Under the common design, the data are sampled at $t_{il} = (2l)/(2m+1)$ with the sample size $n=100$ and the sampling frequency $m=20, 50, 100$, respectively. In each setting, we repeat 100 times independently for $p=50, 100, 200$ and the noise level $\sigma^2=1,3$, respectively. 

For comparison purposes, we estimate eigenvectors by performing conventional PCA on sparse covariance matrices obtained by dynamic covariance models (DCM) \citep{chen2016dynamic}. Since the bandwidth selected by the leave-one-point-out cross-validation in the DCM might be inappropriate for repeated measurements, we use the leave-one-curve-out cross-validation instead and denote the resulting model by DCM+. 
Under the common design, we include the methods of \citet{berrendero2011principal} and \citet{johnstone2009consistency}, denoted by BJS and DT, respectively. Note that the parameters in DT are set to the recommended values in the original paper.
We evaluate the performance of estimators by the mean integrated squared error (MISE) which is approximated by computing the average of squared errors, defined in \eqref{eq:distance}, on a grid of 50 equally spaced points.
Since the BJS and DT only obtain estimates under the sampling locations, to calculate the MISE for these two methods, we simply generate data at evaluated grids to obtain corresponding estimates. Note that the resulting errors $\mathrm{MISE}_{BJS}$ and $\mathrm{MISE}_{DT}$ are not relevant to the sampling frequency.

We begin with illustrating the selection and the performance of tuning parameters in our method.
We set $d=3$ under which the FVE is about 85\%, and other parameters are chosen as discussed in Section \ref{subsec:tuning}. Specifically, the bandwidth $h$ is selected by leave-one-curve-out cross-validation, while the $\rho_t$ and $\gamma_t$ are determined by 5-fold cross-validation to save the computation time. 
Since the quantity $m$ or $\bar{m}$ might be large, to further reduce the computation, we randomly choose 10 observational points with equal probability for each curve to calculate the error in the validation step to tune the bandwidth. The effectiveness of the selection strategy for the parameters is illustrated in Figure \ref{fig:cv}.
As is shown, the selected parameters well depict the true smoothness of covariance matrices, the sparsity level and model complexity of the eigenvectors, respectively.

\begin{table}
	\caption{Average integrated squared errors and standard deviations over 100 replications for different settings under the irregular design and $\sigma^2=3$.}
	\centering
	\begin{tabular}{|c|c|c|c|c|c|}
		\hline
		\multicolumn{2}{|c|}{Model} & $\mathrm{MISE}_0$ & $\mathrm{MISE}$  & $\mathrm{MISE}_{DCM}$ & $\mathrm{MISE}_{DCM+}$ \\ \hline
		\multirow{3}{*}{\begin{tabular}[c]{@{}c@{}}$p$=100 \\ $n$=100 \end{tabular}} & $\bar{m}=100$   & 0.033 (0.014) & 0.031(0.012) & 0.040(0.037) & 0.038(0.031) \\ \cline{2-6}
		& $\bar{m}=50$ & 0.042(0.014) & 0.041(0.016) &  0.128(0.189) & 0.121(0.157) \\ \cline{2-6}
		\ & $\bar{m}=20$ & 0.112(0.056) &  0.102(0.061) & 0.488(0.237) & 0.498(0.226)  \\ \hline
		\multirow{2}{*}{\begin{tabular}[c]{@{}c@{}}$p$=100 \\ $n$=500 \end{tabular}} & $\bar{m}=20$ & 0.023(0.003)
		& 0.022(0.002)   & 0.021(0.006) & 0.022(0.005) \\ \cline{2-6}
		& $\bar{m}=10$ & 0.035(0.008)  & 0.034(0.007)  & 0.059(0.091)  & 0.080(0.130) \\ \cline{2-6}
		& $\bar{m}=4 $ & 0.078(0.016) & 0.068(0.016)  & 0.465(0.158) & 0.460(0.167)  \\ \hline \hline
		\multirow{3}{*}{\begin{tabular}[c]{@{}c@{}}$p$=200 \\ $n$=100 \end{tabular}} & $\bar{m}=100$  &
		0.038(0.023) & 0.036(0.022) & 0.127(0.168) & 0.130(0.177)  \\ \cline{2-6}
		& $\bar{m}=50$  & 0.058(0.037) & 0.059(0.044) & 0.361(0.263) & 0.322(0.275)  \\ \cline{2-6}
		\ &  $\bar{m}=20$  & 0.148(0.071)  & 0.123(0.079) & 0.577(0.193) & 0.574(0.184) \\ \hline
		\multirow{2}{*}{\begin{tabular}[c]{@{}c@{}}$p$=200 \\ $n$=500 \end{tabular}} & $\bar{m}=20$  &  0.024(0.002) & 0.023(0.002)   & 0.095(0.154) & 0.077(0.130)  \\ \cline{2-6}
		& $\bar{m}=10$ & 0.036(0.005) & 0.035(0.005)  & 0.288(0.206) & 0.269(0.214) \\ \cline{2-6}
		& $\bar{m}=4$ & 0.104(0.025)  & 0.080(0.031) & 0.515(0.087) & 0.515(0.086) \\ \hline
	\end{tabular} 
	\label{tab:indpt}
\end{table}

\begin{table}
	\caption{Average integrated squared errors and standard deviations over 100 replications for different settings under the common design and $\sigma^2=3$.}
	\centering
	\begin{tabular}{|c|c|c|c|c|c|c|c|}
		\hline
		\multicolumn{2}{|c|}{Model} & $\mathrm{MISE}_0$  & $\mathrm{MISE}$ & $\mathrm{MISE}_{DCM}$ & $\mathrm{MISE}_{DCM+}$ & $\mathrm{MISE}_{BJS}$ & $\mathrm{MISE}_{DT}$  \\ \hline
		\multirow{3}{*}{\begin{tabular}[c]{@{}c@{}}$p$=100 \\ $n$=100 \end{tabular}} & m=100  & 0.032(0.016)  & 0.030(0.016) & 0.064(0.096) & 0.058(0.065) & \multirow{3}{*}{\begin{tabular}[c]{@{}c@{}} 0.643 \\ (0.075) \end{tabular}} & \multirow{3}{*}{\begin{tabular}[c]{@{}c@{}} 0.647 \\ (0.099) \end{tabular}} \\ \cline{2-6}
		& m=50 & 0.040(0.014) & 0.037(0.013) & 0.212(0.177) & 0.199(0.159) & &  \\ \cline{2-6}
		& m=20  & 0.128(0.081) &  0.114(0.085) & 0.722(0.260) & 0.706(0.276) & & \\  \hline \hline
		\multirow{3}{*}{\begin{tabular}[c]{@{}c@{}}$p$=200 \\ $n$=100 \end{tabular}} & m=100 & 0.037(0.028)  & 0.036(0.025) & 0.242(0.275) & 0.180(0.219) & \multirow{3}{*}{\begin{tabular}[c]{@{}c@{}} 1.021 \\ (0.098) \end{tabular}} & \multirow{3}{*}{\begin{tabular}[c]{@{}c@{}} 0.645 \\ (0.118) \end{tabular}}  \\ \cline{2-6}
		& m=50  & 0.053(0.028) &  0.053(0.030) & 0.581(0.299) & 0.609(0.292) & &  \\ \cline{2-6}
		& m=20  & 0.157(0.063) &  0.132(0.074) & 1.026(0.125) & 1.018(0.147) &  & \\ \hline
	\end{tabular}
	\label{tab:common}
\end{table}

The results for $p$=100 and 200 with $\sigma^2=3$ are summarized in Tables \ref{tab:indpt} and \ref{tab:common}, while the results for $p=50$ are qualitatively similar, thus not reported for space economy. Moreover, the results with $\sigma^2=1$ are provided in the Supplementary Material. The errors of the initial estimate $\hat{U}^0$ and the refined estimate $\hat{U}$ are denoted by $\mathrm{MISE}_0$ and $\mathrm{MISE}$, respectively. 
As Tables \ref{tab:indpt} and \ref{tab:common} show, the proposed method outperforms other methods in all settings, while the refined estimators perform slightly better than the initial estimates under both designs. The DCM and DCM+ methods perform decently when $m=100$ and $p=100$. However, their performance deteriorates significantly if the dimension increases or the sampling frequency becomes small.  Under the irregular design, as noted in Section \ref{subsec:thm-indpt}, the grids can be sparse to obtain consistent estimation as the sample size permits. Thus, even in the very sparse case as $\bar{m}=4$ or 10 ($n=500$), the error is still well controlled and even smaller than that when $\bar{m}=20$ ($n=100$) due to a larger sample size and the advantage of data pooling.  
Moreover, when the total number of observations is comparable, a larger sample size usually leads to a better estimate.
Under the common design, it is not surprising that the BJS method fails to obtain reasonable estimators as it did not accommodate high dimensionality. Although the DT yields sparse eigenvectors, it tends to select too few coordinates, which introduces larger bias. Further, as shown in the left of Figure \ref{fig:raw-smootherror-TNR-TPR}, our method performs significantly better than other methods uniformly over $t$. The results under other simulation settings lead to similar conclusions and are not reported.

\begin{figure}[htbp]
	\centering                       
	\newcommand{\thisgap}{1.5mm}
	\newcommand{\thiswidth}{0.3\linewidth}
	\begin{tabular}{ccc}
		\hspace{\thisgap}\includegraphics[width=\thiswidth]{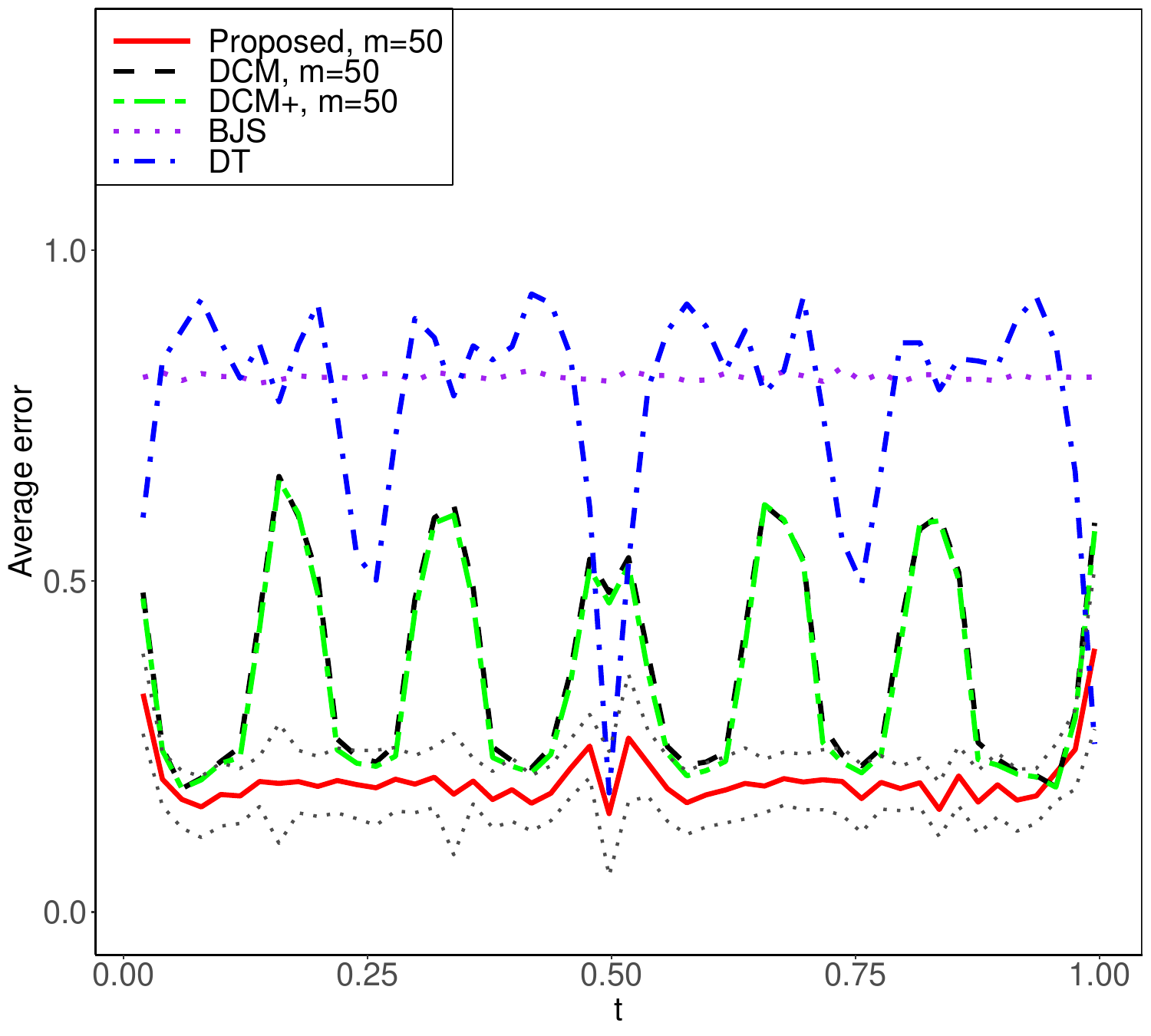} &
		\hspace{\thisgap}\includegraphics[width=\thiswidth]{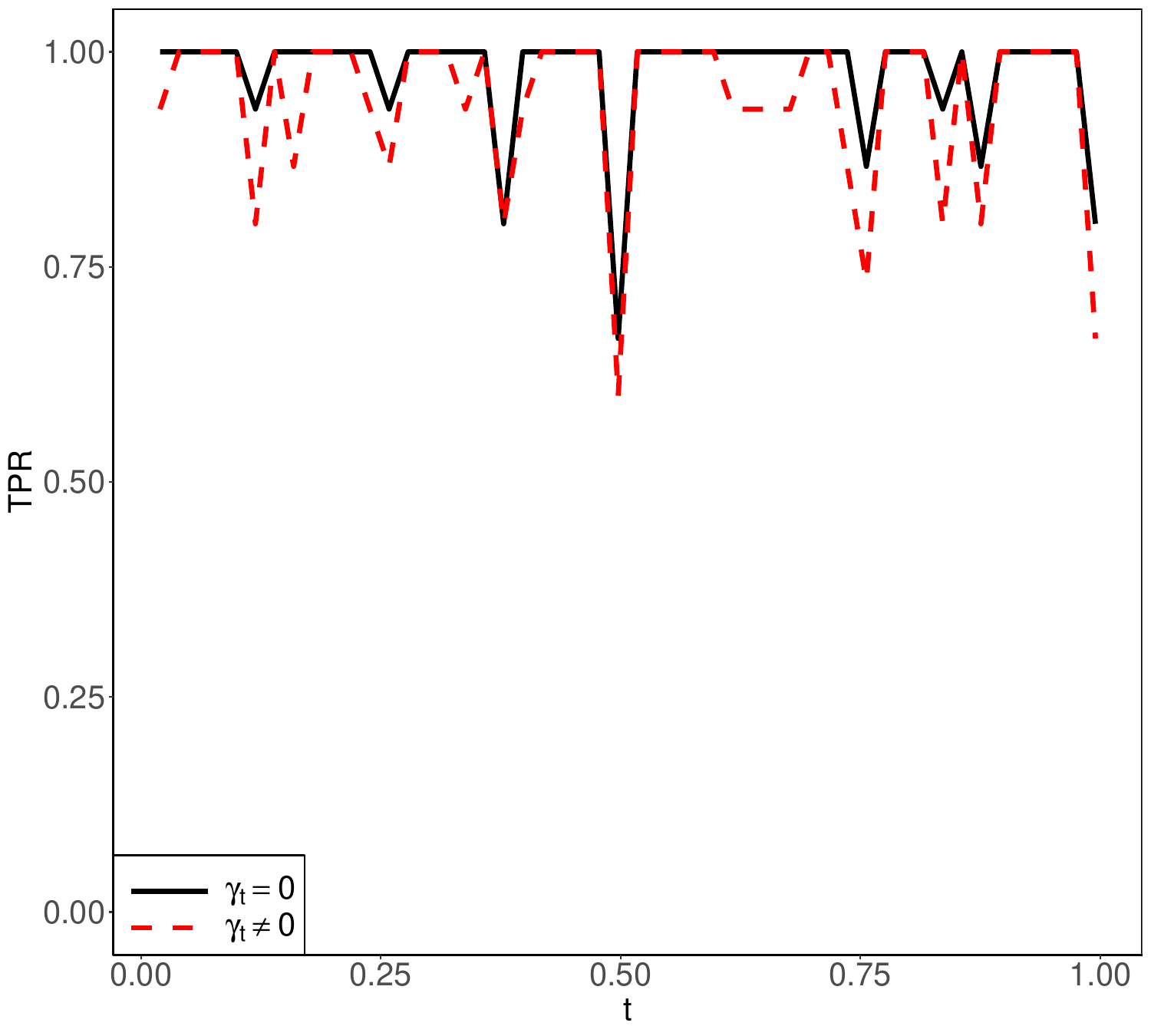} &
		\hspace{\thisgap}\includegraphics[width=\thiswidth]{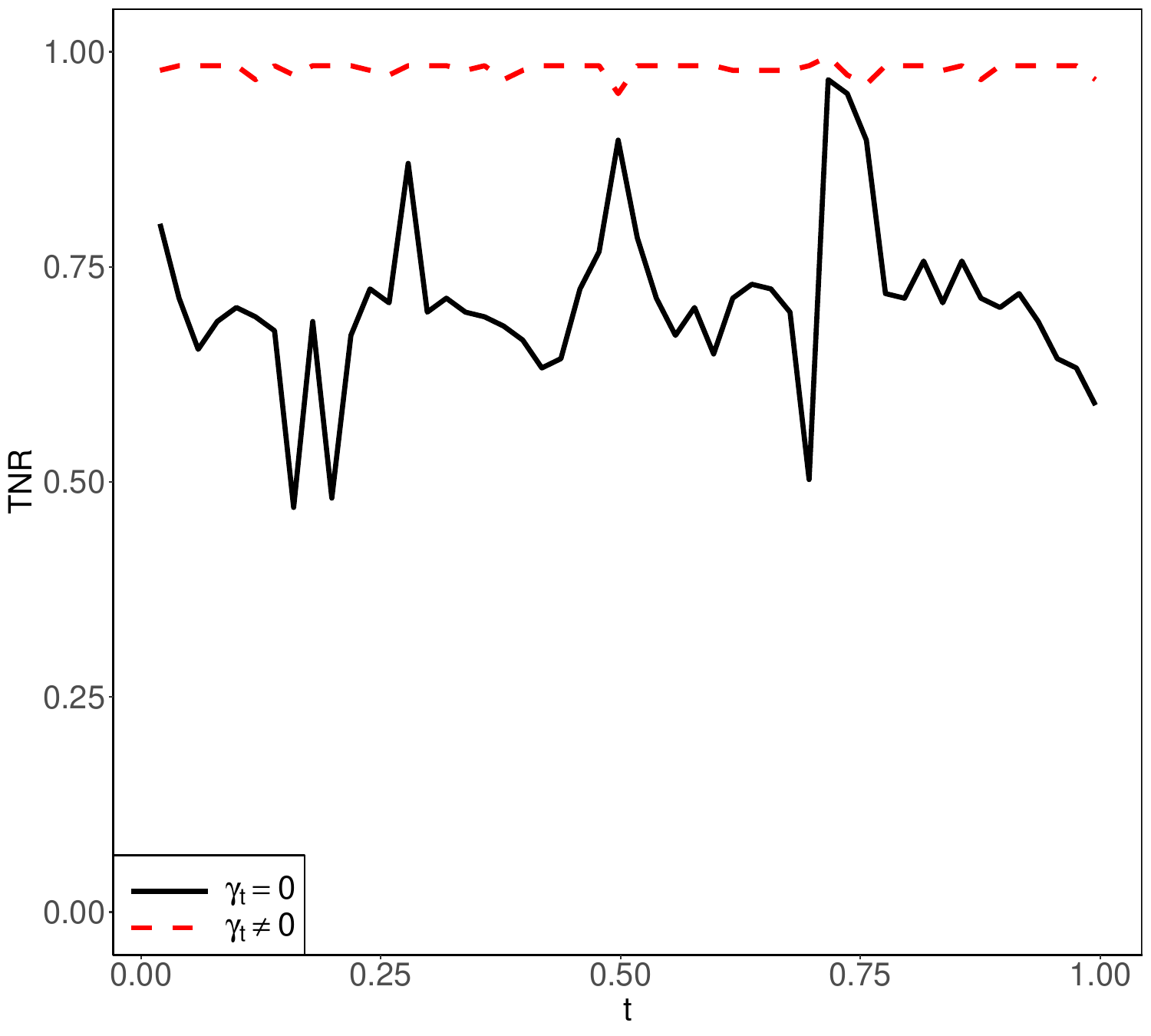} \\
	\end{tabular}
	\vspace{-3mm}
	\caption{The common design: $p=100$ and $\sigma^2=3$. Left: The errors, defined in \eqref{eq:distance}, of different methods over 50 equally spaced points in [0,1]. The two black dotted lines  reflect uncertainty (standard deviation) of errors for the proposed estimate. Right: The performance of TPR (middle) and TNR (right) over $t$ of the proposed estimates with $\gamma_t=0$ (solid) or $\gamma_t\neq 0$ (dashed). It illustrates the advantage of the refinement step in screening out irrelevant variables and achieving desired model parsimony, which improves TNR without decreasing TPR.}
	\label{fig:raw-smootherror-TNR-TPR}
\end{figure}

At last, to illustrate the performance of the refined estimates in achieving model parsimony, two criteria TNR = TN/(TN+FP) and TPR=TP/(TP+FN) are reported in Figure \ref{fig:raw-smootherror-TNR-TPR}, where TP and TN are abbreviations for true positives and true negatives, respectively, i.e., the number of significant or non-significant variables correctly identified by our method, similarly FP and FN stand for false positives and false negatives. 
The numerical values of $\hat{\Pi}_{jj}$ over $10^{-6}$ in magnitude are considered nonzero for the consideration of computation accuracy.
As the principal eigenvectors vary with $t$, the signal $\min_{j\in J(t)} \Pi_{jj}(t)$ may be close to zero which leads to a bit lower TPR at some $t$. 
The result reveals the fact that the refined estimate by thresholding yields better variable selection results, which is particularly useful for model interpretation.

\section{Real data example }\label{sec:realdata}

The heartbeat sound dataset from \url{http://www.timeseriesclassification.com/description.php?Dataset=Heartbeat} were sourced from several contributors around the world, collected at either a clinical or nonclinical environment \citep{liu2016open}.
The heart sound recordings were collected from different locations on the body. The typical four locations are the aortic area, pulmonic area, tricuspid area and mitral area, but could be one of nine different locations. 
Each recording was truncated to 5 seconds. A Spectrogram of each instance was then created with a window size of 0.061 seconds and an overlap of 70\%.
Each instance is arranged such that each dimension is a frequency band from the spectrogram, with $p=61$ and $m=405$.
We focus on extracting the dynamic features of $n=295$ pathological patients.

\begin{table}
	\centering
	\caption{\label{tab:approx} The mean squared recovery errors on the test data with a different number of principal eigenvectors/eigenfunctions for different methods.}
	\begin{tabular}{|c|c|c|c|c|c|c|c|c|}
		\hline
		& $ \mathrm{Proposed}_{0}$ & Proposed & DCM & DCM+ & BJS & DT & Rough & FPCA \\ 
		\hline
		$d=5$ & 0.201 & 0.201 & 0.274 & 0.222 & 0.279 & 0.302 & 0.273 & 3.156 \\ \hline
		$d=6$ & 0.113 & 0.112 & 0.218 & 0.151 & 0.192 & 0.228 & 0.191 & 3.120 \\ \hline
		$d=7$ & 0.072 & 0.072 & 0.187 & 0.114 & 0.132 & 0.184 & 0.131 & 3.078 \\ \hline
		$d=8$ & 0.045 & 0.045 & 0.165 & 0.097 & 0.091 & 0.157 & 0.093 & 3.044 \\ 
		\hline
	\end{tabular}
\end{table}

For the purpose of evaluation, we compute the mean squared recovery error on the held-out test sample. Specifically, we randomly choose 100 subjects as the test data, and treat the remaining data as the training set. The estimators are obtained using the training sample under different methods including the proposed method, DCM, DCM+, BJS and DT. 
To better demonstrate the advantages of smoothing, we compare with a variant of our method with the bandwidth nearly 0, denoted by the rough estimator. More specifically, this variant is obtained by using the same optimization technique in our paper and only replacing the smoothed covariance matrix with the sample covariance matrix. 
The mean squared recovery errors under a different number of eigenvectors $d$ are calculated over all 405 grids, that is, $(100m)^{-1}\sum_{i,l}\|y_{il} - \hat{\mu}_l - \hat{U}_l\hat{U}_l^{\t} (y_{il} - \hat{\mu}_l) \|^2 $, where $y_{il}$ are observations of the $i$-th test subject, $\hat{\mu}_l$ is the sample mean vector and $\hat{U}_l$ is the estimator at the $l$-th observed locations. 
Moreover, to compare with the performance of FPCA \citep{ramsay2005} in terms of the low-dimensional representation, we implement FPCA for each functional variable and calculate the mean squared recovery error, $(100m)^{-1}\sum_{i,j,l} (y_{ijl} - \hat{\mu}_{jl} - \sum_{k=1}^d \hat{\xi}_{jk} \hat{\psi}_{jkl} )^2 $, where $\hat{\xi}_{jl} = m^{-1}\sum_{l=1}^m (y_{ijl} - \hat{\mu}_{jl}) \hat{\psi}_{jkl} $ and $\hat{\psi}_{jkl}$ is the value at the $l$-th time point of the $k$-th eigenfunction for the $j$-th functional variable.
Note that $\mathrm{Proposed}_0$ and Proposed represent the proposed method without and with the refinement step, respectively.
While the refinement step has little effect on the recovery errors, it in fact leads to a more parsimonious model with fewer retained variables, screening out about 20\% insignificant variables. As Table \ref{tab:approx} shows, the proposed method obtains favorable performance over other methods, suggesting more accurate estimation from our approach for the dynamic PCA. In particular, our approach outperforms the rough estimator, showing the usefulness of the smoothing strategy.
The DCM performs worse than DCM+ because it tends to select smaller bandwidth which is not satisfactory in this case.
As seen in Table \ref{tab:approx}, the representation obtained by FPCA is not promising with large recovery errors.

\begin{table}
	\centering
	\caption{\label{tab:approx-irregular} The mean squared recovery errors on the test data for different methods under the irregular case with different $\tilde m$ and $d=6, 8$.}
	\begin{tabular}{|c|c|c|c|c|c|c|c|c|}
		\hline
		\multirow{2}{*}{}& \multicolumn{4}{|c|}{$d=6$} & \multicolumn{4}{|c|}{$d=8$} \\ \cline{2-9}
		& $ \mathrm{Proposed}_{0}$ & Proposed & DCM & DCM+ & $ \mathrm{Proposed}_{0}$ & Proposed & DCM & DCM+ \\ \hline
		$\tilde m=30$ & 0.191 & 0.191 & 0.226 & 0.232 & 0.124 & 0.124 & 0.151 & 0.162 \\ \hline
		$\tilde m=50$  & 0.150 & 0.150 & 0.174 & 0.202 & 0.072 & 0.072 & 0.120 & 0.150 \\ \hline
		$\tilde m=80$ & 0.147 & 0.147 & 0.176 & 0.175 & 0.068 & 0.068 & 0.116 & 0.114 \\ 
		\hline
	\end{tabular}
\end{table}

To demonstrate the performance under the irregular design, we randomly sample $\tilde{m}=30, 50, 80$ measurements, respectively, with equal probability from each subject of the training sample, and use the obtained irregular data for estimation.
Since the methods BJS and DT are not feasible for the irregular design, we compare the recovery errors of the other three approaches for the dynamic PCA. Note that we report the results with $d=6, 8$ in Table \ref{tab:approx-irregular} for space economy, since the results exhibit a similar pattern for other values of $d$. 
It is demonstrated that the proposed method is capable of producing more desirable estimates. Moreover, despite fewer observations, the recovery errors in Table \ref{tab:approx-irregular} are still much lower than the errors of FPCA obtained under the common design.

\section{Concluding remarks}

We propose a unified framework to estimate dynamic eigenvectors in high-dimensional settings by combining the local linear smoothing and the sparsity constraint under both common and irregular designs. The resulting estimators satisfy sparsity and orthogonality simultaneously. Different from the conventional nonparametric smoothing, the rates of convergence depend on the sampling frequency and the sample size jointly, exhibiting the phase transition phenomenon. When the sampling frequency is suitably large, the obtained rates are optimal as if the whole curves are available under both designs. Otherwise, the irregular design is preferred with a faster rate of convergence. 

It is interesting to study other types of smoothing techniques, such as smoothing splines, for the problem of DPCA. Moreover, since the PCA is sensitive to outliers, it is also useful to develop a dynamic robust model. These topics are beyond the scope of the current paper and deserve future study.

\section*{Appendix}
\begin{appendix}
\setcounter{equation}{0}
\renewcommand{\theequation}{\Alph{section}.\arabic{equation}}

\section{Auxiliary lemmas}\label{apx:proofs}

In the sequel, we suppress the index $t$ of pointwise results for convenience when no ambiguity arises. 
We write $a \asymp b$ if $a \lesssim b$ and $b \lesssim a$ hold simultaneously.

\begin{lemma}\label{lem:dpca}
	Given the random function $\bX(t) \in \real^p$ with the mean function $\bmu(t) = E \bX(t)$, $t \in \tdomain$. Denote $U(t) = (\bu_1(t), \dots, \bu_d(t))$, where $\bu_1(t), \dots, \bu_d(t)$ are the first $d$ eigenvectors of $\bX(t)$. Then, $U(t)$ is a solution of the optimization problem \eqref{opt:dpca}.
	Moreover, solving \eqref{opt:dpca} is reduced to performing multivariate PCA at each $t$.
\end{lemma}

\begin{proof}
	Let $h\big(V(t)\big) = E\|\bX(t) - \bmu(t) - V(t)V(t)^{\t}\{\bX(t)-\bmu(t)\}\|^2$. We first show that $U(t)$ is a solution of \eqref{opt:dpca}, and define the optimization problem for multivariate PCA at each $t$ as,
	\begin{eqnarray}\label{eq:pca}
	\min_{V(t)} & h\big(V(t)\big) \nonumber\\
	s.t. & V(t)^{\t} V(t) = I_d.
	\end{eqnarray}
	Denote by $U^*(t)$ the solution of \eqref{eq:pca}, and let $h^*(t) = h\big( U^*(t) \big)$. Moreover, we denote the solution of \eqref{opt:dpca} by $U^{**}(t)$, and let $h^{**}(t) = h\big( U^{**}(t) \big)$. Thus, we have $\int_{\tdomain} \{h^*(t) - h^{**}(t) \} dt \le 0 $. Since $U^*(t)$ is feasible for \eqref{opt:dpca}, then $\int_{\tdomain} h^{**}(t) dt \le \int_{\tdomain} h^*(t)dt$.
	Due to the fact that $\int_{\tdomain} h^*(t) dt = \int_{\tdomain} h^{**}(t) dt$, we conclude that $U^*(t)$ is a solution of \eqref{opt:dpca}. Note that (A.1) seeks an orthonormal matrix $V(t) \in \real^{p \times d}$ to minimize $E\|\bX(t) - \bmu(t) - V(t)V(t)^{\t}\{\bX(t)-\bmu(t)\}\|^2$, which is equivalent to finding an orthonormal matrix to maximize $\mathrm{Tr}(V(t)^{\t}\Sigma(t)V(t))$. Thus $U(t)$ is a solution of \eqref{eq:pca}, then it is also a solution of \eqref{opt:dpca}.  
	
	Next, we show that \eqref{opt:dpca} is reduced to the multivariate PCA at each $t$. 
	If $U^{**}(t)$ is not the solution of \eqref{eq:pca}, then there exists some $\tilde{t} \in \tdomain$ such that the columns of $U^{**}(\tilde{t})$ does not correspond to the first $d$ eigenvectors of $\bX(\tilde{t})$. It contradicts the argument that $U^{**}(t)$ is a solution of \eqref{opt:dpca} since replacing $U^{**}(\tilde{t})$ with $U(\tilde{t})$ leads to a smaller objective value for \eqref{opt:dpca}.
\end{proof}

\begin{lemma}\label{lemma:error_varinfty}
	Recall that $U$ and $\hat{U}$ are true and estimated principal eigenvectors with projection matrices $\Pi=UU^{\t}$ and $\hat{\Pi}=\hat{U}\hat{U}^{\t}$, respectively. Assume $U \in \mathcal{U}(q, R_q)$, $0 \le q \le 1$. If conditions in Lemma \ref{lemma:support} hold, with appropriate choice of parameters, we have $d\{U, \hat{U}\} \leq C_q\|\hat{\Sigma} - \Gamma\|_{\infty}^{1-q/2}$ where $C_q = Cd^2R_q$ for some positive constant $C>0$.
\end{lemma} 

\begin{lemma}\label{lemma:w_R_rate}
	If $h\to 0$ and $n\bar{m}h \to \infty$, then under the irregular design, we have
	\begin{itemize}
		\item[(a)] $R_{\ell} \asymp n\bar{m}h^{\ell}\big(1+o_p(1)\big)$, $\ell=0, 1, 2$. Moreover, $R_2R_0 - R_1^2 \asymp n^2 \bar{m}^2 h^2 \big (1+o_p(1)\big)$.
		\item[(b)] $E\left[\left\{ R_2K_h(t_{il}-t) - R_1K_h(t_{il}-t)(t_{il}-t)\right\} x_{ijl}x_{ikl} \right]^2 = O(n^2\bar{m}^2h^3)$.
		\item[(c)] $E \left( \tilde{w}_{il} x_{ijl}x_{ikl}  \tilde{w}_{i'l'} x_{i'jl'}x_{i'kl'} \right) = O(n^2\bar{m}^2h^4) $ for $(i,l) \neq (i', l')$, where $\tilde{w}_{il} = R_2K_h(t_{il}-t) - R_1K_h(t_{il}-t)(t_{il}-t)$.
	\end{itemize}
\end{lemma}

\begin{lemma}\label{lemma:apxmaxvar}
	Under Assumptions \ref{assump:tail}-\ref{assump:indpt-m}, we have for each $t \in \tdomain$,
	\[ \max_{j,k} \left|\sum_{i=1}^{n}\sum_{l=1}^{m_i}\left\{ \tilde{w}_{il} x_{ijl}x_{ikl} - E\left(\tilde{w}_{il}x_{ijl}x_{ikl}\right) \right\}\right| = O_p\{ (\log p)^{1/2} (n^3\bar{m}^3h^3 + n^3\bar{m}^4h^4)^{1/2}\}. \]
\end{lemma}

The proofs of Lemmas \ref{lemma:error_varinfty}-\ref{lemma:apxmaxvar} are deferred to the Supplementary Material. In the following, we provide the proof of Theorem \ref{thm:indpt}, while the proof of Theorem \ref{thm:common} is analogous which could be found in the Supplementary Material.

\section{Proofs of main results}

\begin{proof}[Proof of Theorem \ref{thm:indpt}]
	
	From Lemma \ref{lemma:error_varinfty}, we have $d\{ \hat{U}(t), U(t) \} \le C_q \|\hat{\Sigma}(t) - \Sigma(t) - \sigma^2 I_p \|_{\infty}^{1-q/2}$. Thus, it suffices to quantify the error $\|\hat{\Sigma}(t) - \Sigma(t) - \sigma^2 I_p \|_{\infty}$.
	
	Let $\Omega(t) = E(\bX(t)\bX(t)^{\t})$, we have $\Sigma(t) = \Omega(t) - \bmu(t) \bmu(t)^{\t}$. Under the irregular design, using the triangle inequality, 
	\begin{eqnarray} \label{eq:|S-Sigma|}
	& & \|\hat{\Sigma}(t) - \Sigma(t) - \sigma^2I_p\|_{\infty} \nonumber\\
	& = & \|\sum_{i=1}^n \sum_{l=1}^{m_i} w_{il} \by_{il}\by_{il}^{\t}-\sum_{i=1}^n \sum_{l=1}^{m_i} w_{il} \by_{il} \sum_{i=1}^n \sum_{l=1}^{m_i} w_{il} \by_{il}^{\t}-\{\Omega(t) - \bmu(t)\bmu(t)^{\t} \} - \sigma^2 I_p \|_{\infty} \nonumber\\
	& \leq & \bigg \|\sum_{i=1}^n \sum_{l=1}^{m_i} w_{il} \by_{il}\by_{il}^{\t} - \Omega(t) - \sigma^2 I_p \bigg\|_{\infty} + \bigg \| \sum_{i=1}^n \sum_{l=1}^{m_i} w_{il} \by_{il} \sum_{i=1}^n \sum_{l=1}^{m_i} w_{il} \by_{il}^{\t} - \bmu(t)\bmu(t)^{\t} \bigg \|_{\infty} \nonumber\\
	& = & M_1(t) + M_2(t).
	\end{eqnarray}
	Note that 
	\begin{eqnarray}\label{eq:M_1}
	& & M_{1}(t)  =  \|\sum_{i=1}^n \sum_{l=1}^{m_i} w_{il} (\bx_{il}+ \bepsilon_{il} )(\bx_{il}+\bepsilon_{il})^{\t} - \Omega(t) -\sigma^2 I_p  \|_{\infty} \nonumber\\
	&\leq &  \bigg\| \sum_{i=1}^n \sum_{l=1}^{m_i} w_{il} \bx_{il}\bx_{il}^{\t} - \Omega(t) \bigg\|_{\infty}  + 2\bigg \|\sum_{i=1}^n \sum_{l=1}^{m_i} w_{il} \bx_{il}\bepsilon_{il}^{\t} \bigg\|_{\infty}  + 
	\bigg \|\sum_{i=1}^n \sum_{l=1}^{m_i} w_{il} \bepsilon_{il}\bepsilon_{il}^{\t} - \sigma^2 I_p \bigg\|_{\infty},
	\end{eqnarray}
	where $\bepsilon_{il} = (\epsilon_{i1l}, \dots, \epsilon_{ipl})^{\t}$. Similarly,
	\begin{eqnarray}\label{eq:M_2}
	M_2(t) & \leq &  \left\|\sum_{i=1}^n \sum_{l=1}^{m_i} w_{il} \bx_{il} \sum_{i=1}^n \sum_{l=1}^{m_i} w_{il} \bx_{il}^{\t} - \bmu(t)\bmu(t)^{\t} \right\|_{\infty} \nonumber\\
	& & + 2\left\|\sum_{i=1}^n \sum_{l=1}^{m_i} w_{il} \bx_{il} \sum_{i=1}^n \sum_{l=1}^{m_i} w_{il} \bepsilon_{il}^{\t} \right\|_{\infty} + \left\|\sum_{i=1}^n \sum_{l=1}^{m_i} w_{il} \bepsilon_{il} \sum_{i=1}^n \sum_{l=1}^{m_i} w_{il} \bepsilon_{il}^{\t} \right\|_{\infty}.
	\end{eqnarray}
	
	To bound the term $\Delta_n = \| \sum_{i=1}^n \sum_{l=1}^{m_i} w_{il} \bx_{il}\bx_{il}^{\t} - \Omega(t)\|_{\infty}$, we have
	\begin{eqnarray*}
		& & \Delta_n  =  \bigg\| \sum_{i=1}^n\sum_{l=1}^{m_i} w_{il}\bx_{il}\bx_{il}^{\t} - \Omega(t) \bigg\|_{\infty}\\ 
		&= & \max_{j,k}
		\bigg| \frac{\sum_{i=1}^n\sum_{l=1}^{m_i} \left\{ \tilde{w}_{il} x_{ijl}x_{ikl} - E(\tilde{w}_{il}x_{ijl}x_{ikl}) \right\}}{R_0R_2 - R_1^2} + \frac{\sum_{i=1}^n\sum_{l=1}^{m_i} \left\{ E(\tilde{w}_{il}x_{ijl}x_{ikl}) - \tilde{w}_{il}\omega_{jk}(t)\right\}}{R_0R_2 - R_1^2} \bigg| \\
		& \le & 
		\frac{\max_{j,k} \bigg|  \sum_{i=1}^n\sum_{l=1}^{m_i} \left\{ \tilde{w}_{il} x_{ijl}x_{ikl} - E(\tilde{w}_{il}x_{ijl}x_{ikl}) \right\} \bigg|}{ \bigg| R_0R_2 - R_1^2 \bigg|} +   \frac{\max_{j,k} \bigg| \sum_{i=1}^n\sum_{l=1}^{m_i} \left\{ E(\tilde{w}_{il}x_{ijl}x_{ikl}) - \tilde{w}_{il}\omega_{jk}(t)\right\} \bigg|}{\bigg| R_0R_2 - R_1^2 \bigg|} \\
		& = & I + II,
	\end{eqnarray*}
	where $\tilde{w}_{il} = R_2K_h(t_{il}-t) - R_1K_h(t_{il}-t)(t_{il}-t)$ and $\omega_{jk}(t) = E\{x_{ij}(t)x_{ik}(t)\}$. The second equality holds due to the fact that $\sum_{i=1}^n \sum_{l=1}^{m_i}w_{il}=1$.  
	
	Denote $\zeta_1 = \max_{j,k} \bigg|  \sum_{i=1}^n\sum_{l=1}^{m_i} \left\{ \tilde{w}_{il} x_{ijl}x_{ikl} - E(\tilde{w}_{il}x_{ijl}x_{ikl}) \right\} \bigg|$. By Lemma \ref{lemma:apxmaxvar}, we conclude $\zeta_1 = O_p\{ (\log p)^{1/2} (n^3\bar{m}^3h^3 + n^3\bar{m}^4h^4)^{1/2}\}$. From (a) of Lemma \ref{lemma:w_R_rate}, we have $R_0 R_2 - R_1^2 \asymp n^2\bar{m}^2h^2\big(1+o_p(1)\big)$. Consequently,
	\[I = O_p\left\{\left(\frac{\log p}{n\bar{m}h} + \frac{\log p}{n}\right)^{1/2}\right\}. \]
	
	Next we bound the term $II$. Denote that $\zeta_2 = \max_{j,k} \bigg| \sum_{i=1}^n\sum_{l=1}^{m_i} \left\{ E(\tilde{w}_{il}x_{ijl}x_{ikl}) - \tilde{w}_{il}\omega_{jk}(t)\right\} \bigg|$. Notice that
	\begin{eqnarray*}
		\zeta_2 & = & \max_{j,k} \bigg| \sum_{i=1}^n\sum_{l=1}^{m_i} \left\{E(\tilde{w}_{il}\omega_{jk}(t_{il})) - \tilde{w}_{il}\omega_{jk}(t)  \right\} \bigg| \\
		& = & \max_{j,k} \bigg| \sum_{i=1}^n\sum_{l=1}^{m_i} \left(E\left[\tilde{w}_{il}\left\{\omega_{jk}(t) + \omega_{jk}^{(1)}(t)(t_{il}-t) + \frac{\omega_{jk}^{(2)}(\xi_{il})}{2}(t_{il}-t)^2\right\}\right] - \tilde{w}_{il}\omega_{jk}(t)  \right) \bigg| \\
		& \le & \max_{j,k} |\omega_{jk}(t)| \bigg| \sum_{i=1}^n\sum_{l=1}^{m_i} (\tilde{w}_{il} - E \tilde{w}_{il}) \bigg| + \bigg| \sum_{i=1}^n\sum_{l=1}^{m_i} E \left\{\frac{\tilde{w}_{il}\omega_{jk}^{(2)}(\xi_{il})}{2}(t_{il}-t)^2 \right\} \bigg| \\
		& = & II_1 + II_2,
	\end{eqnarray*}
	where $\xi_{il}$ is between $t$ and $t_{il}$, and the inequality holds since $\sum_{i=1}^n\sum_{l=1}^{m_i} \tilde{w}_{il}(t_{il} - t) = 0$. 
	Using similar arguments for the proof of Lemma \ref{lemma:w_R_rate}(b), we obtain $II_1 = O_p\{(n\bar{m}h)^{3/2}\}$. To bound the term $II_2$, notice that $\omega_{jk}^{(2)}(\xi_{il})$ is bounded by Assumption \ref{assump:smoothmean-cov} and 
	\begin{eqnarray*}
		E\{\tilde{w}_{il} (t_{il}-t)^2\} 
		& = & E\left[\left\{ R_2 K_h(t_{il}-t) - R_1 K_h(t_{il}-t)(t_{il}-t) \right\} (t_{il}-t)^2\right].
	\end{eqnarray*}
	Note that
	\begin{eqnarray*}
		E\left\{ R_2 K_h(t_{il}-t) (t_{il}-t)^2 \right\}
		& = & E \left[ \left\{ \sum_{i=1}^n \sum_{l=1}^{m_{i}} K_h(t_{il}-t)(t_{il}-t)^2\right\} K_h(t_{il}-t)(t_{il}-t)^2\right] \\
		& = & O(n\bar{m}h^4),
	\end{eqnarray*}
	by the change of variables. Analogously, we show that $E \left[ \left\{R_1 K_h(t_{il}-t)(t_{il}-t) \right\} (t_{il}-t)^2 \right] = O(n\bar{m}h^4)$. Thus, $II_2 = O(n^2\bar{m}^2h^4)$. 
	According to Lemma \ref{lemma:w_R_rate}(a), we have $R_0R_2 - R_1^2 = n^2\bar{m}^2h^2\big(1+o_p(1)\big)$. Combining these pieces together leads to the fact that $II=O_p\{h^2 + 1/(n\bar{m}h)^{1/2}\}$.
	
	The rates of other terms are proved using similar arguments which are omitted here to save space. By \eqref{eq:|S-Sigma|}, \eqref{eq:M_1} and \eqref{eq:M_2}, we obtain
	\begin{eqnarray*}
		\|\hat{\Sigma}(t) - \Gamma(t)\|_{\infty} =  O_p\left\{\left(\frac{\log p}{n\bar{m}h} + \frac{\log p}{n}\right)^{1/2} + h^2\right\} ,
	\end{eqnarray*}
	which completes the proof together with Lemma \ref{lemma:error_varinfty}.
\end{proof}
\end{appendix}

\bibliographystyle{asa}
\bibliography{DPCA}
\end{document}